\documentclass[%
 aip,
 cha,
 amsmath,amssymb,
 reprint,%
]{revtex4-1}

\usepackage{graphicx}
\usepackage{caption}
\usepackage{subcaption}
\usepackage{amsmath, amssymb, graphicx, setspace}
\usepackage{placeins}
\allowdisplaybreaks
\usepackage{dcolumn}
\usepackage{bm}
\usepackage{chngcntr}
\usepackage{amsthm}
\usepackage{blkarray}
\usepackage{multirow}
\usepackage{mathtools}

\usepackage[utf8]{inputenc}
\usepackage[T1]{fontenc}
\usepackage{mathptmx}
\usepackage{etoolbox}

\newtheorem{theorem}{Theorem}[section]
\newtheorem{corollary}{Corollary}[section]
\newtheorem{lemma}{Lemma}[section]
\newtheorem{proposition}{Proposition}[section]

\newtheorem{definition}{Definition}[section]

\newtheorem*{example}{Example}

\newcommand{\mathsym}[1]{{}}
\newcommand{\unicode}[1]{{}}

\makeatletter
\def\@email#1#2{%
 \endgroup
 \patchcmd{\titleblock@produce}
  {\frontmatter@RRAPformat}
  {\frontmatter@RRAPformat{\produce@RRAP{*#1\href{mailto:#2}{#2}}}\frontmatter@RRAPformat}
  {}{}
}%
\makeatother
\begin{document}

\preprint{AIP/123-QED}

\title[Complete Synchronization of Coupled Oscillators Based on Contraction Theory]{Complete Synchronization of Coupled Oscillators Based on Contraction Theory}
\author{Brian Y Zhang}
\email{bzhang8@nd.edu.}
\affiliation{ 
	Physics Department, University of Notre Dame
}%
\author{Masoud Asadi-Zeydabadi}%
\email{masoud.asadi-zeydabadi@ucdenver.edu.}

\author{Randall Tagg}
\email{randall.tagg@ucdenver.edu.}
\affiliation{%
	Physics Department, University of Colorado Denver
}%

\date{\today}

\begin{abstract}
	This paper studies contraction theory with the aim of exploring complete synchronization phenomenon in complex networks of coupled oscillators. We examine the conditions for complete synchronization in three network topologies: all-to-all, star, and ring. Specifically, we derive the conditions under which networks of linearly coupled (Duffing) van der Pol oscillators achieve complete synchronization. Additionally, by combining contraction theory with the trapping region method, we identify the conditions for complete synchronization of networks of linearly coupled Rayleigh van der Pol oscillators under specific initial conditions.
\end{abstract}

\maketitle

\begin{quotation}
The study of coupled oscillators and complex systems spans various fields, including mathematics, engineering, robotics, and biology\cite{mirollo,roy,golubitsky,winfree}. Many of these complex systems can be described by networks, where nodes symbolize individual oscillators and links represent the couplings between them. The advancement of network science has significantly enhanced our understanding of complex systems. Synchronization is a key area of interest in this field, as it helps explain the underlying mechanisms of diverse collective behaviors in complex systems\cite{pikovsky,ge}.
\end{quotation}

\section{Introduction}

Early studies observed synchronization phenomena in a variety of artificial devices. As research progressed, it became evident that synchronization is also prevalent in natural phenomena, such as Josephson junctions, nanomechanics, neurodynamics, the synchronous flickering of fireflies, the collective chirping of crickets, and recoil atomic lasers\cite{wiesenfeld,cross,nair,ermentrout,walker,javaloyes}. Understanding these self-organized cooperative states is crucial not only for comprehending group dynamics in complex systems but also for conducting relevant experiments and exploring potential applications\cite{montbrio,xu}.

\textit{Complete synchronization}, where both phase and amplitude are synchronized, is a significant aspect of this phenomenon. Contraction theory, rooted in fluid dynamics and differential geometry, has proven to be an effective tool for analyzing complete synchronization behaviors in nonlinear networks\cite{simpson,lohmiller}. Constructing a \textit{virtual system} of the network and identifying its \textit{contraction region} makes it possible to predict whether the network will achieve complete synchronization theoretically.

\textit{Van der Pol} (vdP) oscillators, initially developed to model electrical circuits employing vacuum tubes, are renowned for their ability to exhibit self-sustained oscillations\cite{van}. The dynamics of coupled vdP oscillators have been extensively studied due to their nonlinear damping properties\cite{kashchenko,wang2024}. Generalizations of vdP oscillators, such as \textit{Duffing van der Pol} (DvdP) and \textit{Rayleigh van der Pol} (RvdP) oscillators, offer richer dynamics, including bifurcation and chaos\cite{zhang2023,szemplinska}, and are particularly useful in modeling the chaotic and complex coordinated behaviors in real-world systems\cite{buslowicz,haken} and an electronic Central Pattern Generator that produces biped gait patterns for robotic systems \cite{de}. Using contraction theory, we analyze the conditions under which these coupled generalized vdP oscillators achieve complete synchronization. We also verify our results using numerical simulations.

To address the issue of networks without virtual systems, we introduced the concept of a \textit{virtual network}. This allows for constructing a higher-dimensional ``virtual system", enabling the identification of synchronization conditions for the original network. A critical challenge in applying contraction theory is ensuring that all oscillator trajectories in a network remain within the contraction region of the corresponding virtual system. For specific coupled-oscillator networks (e.g. coupled vdP and DvdP oscillators), adjusting coefficients can extend the contraction region to cover the entire phase plane, ensuring all trajectories remain within this region. However, for most complex networks (e.g. coupled RvdP oscillators), the contraction regions are finite regardless of coefficients. Our approach involves finding a trapping region within the contraction region of the virtual system, ensuring that trajectories starting within the trapping region remain within the contraction region.

This paper is organized as follows: Section 2 introduces basic contraction theory. Section 3 discusses the application of contraction theory to complete synchronization in various network topologies, including all-to-all, star, and ring topologies. Section 4 examines the dynamics of linearly coupled vdP, DvdP, and RvdP oscillators, providing the conditions for each to achieve complete synchronization.

In this paper, we will use the following notations: ``\(M \prec\ 0\) (\(\succ 0\))" means that the matrix \(M\) is negative (positive) definite; ``\(M \preceq 0\) (\(\succeq 0\))" means that \(M\) is negative (positive) semi-definite. 

\section{Contraction Theory}
\begin{definition}\label{def2.1}(\pmb{Contracting Region})
	Given a system equation of the form
	
	\begin{equation} \label{1}
		\dot{\pmb{x}}=\pmb{f}(\pmb{x}(t)),
	\end{equation}
	
	\noindent where \(\pmb{x}\in \mathbb{R}^n\) is a set of state variables and \(\pmb{f}:\mathbb{R}^n\to \mathbb{R}^n\) is a smooth nonlinear vector function. 
	A contraction region \(\mathfrak{C}\) of this system is a region of the state space where the symmetric part of the Jacobian \(\left(\frac{\partial f}{\partial \pmb{x}}\right)_s\) is uniformly negative definite, i.e., \(\mathfrak{C}\text{:=}\left\{\pmb{x}\in \mathbb{R}^n:\lambda _{\max } (\pmb{x}) <0\right\}\), where \(\lambda _{\max }(\pmb{x})\) is the largest eigenvalue of \(\frac{1}{2} \left(\frac{\partial \pmb{f}}{\partial \pmb{x}}+\left(\frac{\partial \pmb{f}}{\partial
		\pmb{x}}\right)^T\right)\).
\end{definition}

\begin{theorem}\label{thm2.1}
	If any two trajectories, \(\pmb{x}_1(t)\) and \(\pmb{x}_2(t)\), of a system of the form (\ref{1}), starting from different initial conditions, remain within \(\mathfrak{C}\), then they converge exponentially
	to each other.
\end{theorem}

\begin{proof}
	See Ref. \onlinecite{lohmiller}.
\end{proof}

In order to apply this theorem to the synchronization of coupled oscillators, the concept of the \textit{virtual system} needs to be introduced\cite{wang,ruzitalab}: a system expressed in the state variables of the original network of the coupled oscillators and a set of new virtual state
variables so it can recover the trajectory of each node in the original network by substituting the state variables of each node
for the virtual state variables. By constructing a virtual system, the multiple coupled systems are {``}condensed{''} into a single system
as described in Eq.~(\ref{1}), to which the contraction theory can be applied.

\section{Complete Synchronization in different network topologies}
Consider a network of nonlinear systems

\begin{equation} \label{2}
	\dot{\pmb{x}}_i=\pmb{f}_i\left(\pmb{x}_i\right)+\pmb{u}_i, i=1,2,\cdots ,N,
\end{equation}

\noindent where \(\pmb{x}_i=\left({x_{i}}_1 \; {x_{i}}_2 \; \cdots \; {x_{i}}_n\right){}^T\in \mathbb{R}^n\) is the state variables of node \(i\),
\(\pmb{f}_i:\mathbb{R}^n\to \mathbb{R}^n\) are nonlinear smooth vector functions, \(\pmb{u}_i\) are coupling
functions which depend on the difference between the state variables of \(i^{\text{th}}\) node and those of others that are linked with it.

\begin{definition} \label{def3.1}(\pmb{complete synchronization})
	The network described by (\ref{2}) is said to be completely synchronized if \(\forall j= 1,
	\cdots ,n\), \(\exists \pmb{x}_i(0)\in \mathbb{R}^n\) for \(i=1,\cdots ,N\) such that
	
	\[\lim_{t\to \infty }  x_{1j}(t)=\lim_{t\to \infty }  x_{2j}(t)=\cdots =\lim_{t\to \infty }  x_{N j}(t)\neq 0.\]
\end{definition}

\begin{theorem}\label{thm3.1}
	Consider a network of the form (\ref{2}). Assume that there exists a virtual system of this network:
	\begin{equation} \label{3}
		\dot{\pmb{y}}=\pmb{\Phi} \left(\pmb{y},\pmb{x}_1,\pmb{x}_2,\text{...},\pmb{x}_N\right),
	\end{equation}
	
	\noindent where \(\pmb{y}\in \mathbb{R}^n\) is the virtual state variable, such that
	\begin{equation*}
		\begin{split}
			\Phi \left(\pmb{x}_1,\pmb{x}_1,\pmb{x}_2,\text{...},\pmb{x}_N\right)&=\pmb{f}_1\left(\pmb{x}_1\right)+\pmb{u}_1,     \\
			\Phi \left(\pmb{x}_2,\pmb{x}_1,\pmb{x}_2,\text{...},\pmb{x}_N\right)&=\pmb{f}_2\left(\pmb{x}_2\right)+\pmb{u}_2,             \\
			&\vdots     \\
			\Phi \left(\pmb{x}_N,\pmb{x}_1,\pmb{x}_2,\text{...},\pmb{x}_N\right)&=\pmb{f}_N\left(\pmb{x}_N\right)+\pmb{u}_N.
		\end{split}
	\end{equation*}
	
	\noindent Then if the trajectories of the oscillators, \(\pmb{x}_1(t),\pmb{x}_2(t),\text{...},\pmb{x}_{N}(t)\), remain within the contraction region of the virtual system (\ref{3}) for some initial conditions, they are completely synchronized.
\end{theorem}

\begin{proof}
	See Ref. \onlinecite{wang}.
\end{proof}

In this paper, we only consider identical oscillators with linear couplings. Note that for a non-identical two-coupled-oscillator system with linear couplings, one may construct its virtual system as \small

\begin{equation*}
	\dot{\pmb{y}}=(\pmb{y}-\pmb{x}_2)\frac{\pmb{f}_1\left(\pmb{y}\right)}{\pmb{x}_1-\pmb{x}_2}+(\pmb{y}-\pmb{x}_1)\frac{\pmb{f}_2\left(\pmb{y}\right)}{(\pmb{x}_2-\pmb{x}_1)}+A_1\pmb{x}_2+A_2\pmb{x}_1-(A_1+A_2)\pmb{y}.
\end{equation*}

\noindent \normalsize where \(A_1\) and \(A_2\) are oscillators 1 and 2 coupling matrices, respectively. However, this system forms singularities when the trajectories of the two oscillators coincide, and thus the right-hand-side function is not smooth. Therefore, the contraction theory is generally not applicable to the coupled non-identical oscillators. The coupled non-identical oscillators cannot generally be completely synchronized at an exponential rate. 

Not all networks have virtual systems of the form (\ref{3}). We require the networks to possess some symmetry. In the following, we discuss three common network topologies --- all-to-all topology, star topology, and ring topology.

\subsection{All-to-all Topology}
\begin{proposition}\label{prop3.1}
	Consider \(N\) all-to-all coupled identical oscillators with symmetric linear couplings described by
	\begin{equation} \label{5*}
		\dot{\pmb{x}}_i=\pmb{f}\left(\pmb{x}_i\right)+\sum_{j=1}^{N} A_j(\pmb{x}_j-\pmb{x}_i), \forall i =1,\cdots
		,N,
	\end{equation} 
	Here, the ``symmetric coupling" means that each oscillator has the same effect on all other oscillators. This network is completely synchronized if \(\pmb{x}_{i}(t)\) are always inside the region where \({J_{\pmb{f}}}_s-\sum_{j=1}^{N} {A_{j}}_s \prec 0\) for all \(i\), where \({J_{\pmb{f}}}_s\) refers to the symmetric part of the Jacobian of \(\pmb{f}\).
\end{proposition}

\begin{proof}
	See Appendix \ref{appA:subsec1}.
\end{proof}

\subsection{Star Topology}
\begin{proposition}\label{prop3.2}
	Consider \(N\) identical oscillators coupled by a star network with symmetric linear couplings described by
	\begin{subequations} \label{8}
		\begin{align}
			\dot{\pmb{x}}_1&=\pmb{f}\left(\pmb{x}_1\right)+\sum_{j=2}^{N} A_j(\pmb{x}_j-\pmb{x}_1), \label{8a}\\
			\dot{\pmb{x}}_i&=\pmb{f}\left(\pmb{x}_i\right)+A_1(\pmb{x}_1-\pmb{x}_i), \text{  }\forall i =2,\cdots \label{8b}
			,N,
		\end{align}
	\end{subequations}
	This network is completely synchronized if \(\pmb{x}_1(t)\) and \(\pmb{x}_2(t)\) are always inside the region where \({J_{\pmb{f}}}_s-\sum_{j=1}^{N} {A_{j}}_s \prec 0\) and \(\pmb{x}_{i}(t)\) are always inside the region where \({J_{\pmb{f}}}_s-{A_{1}}_s \prec 0\) for all \(i\) from 2 to \(N\).
\end{proposition}

\begin{proof}
	See Appendix \ref{appA:subsec2}.
\end{proof}

The limitation of applying the above propositions is that when the symmetric part of the Jacobian of the virtual system of a network is not uniformly negative definite on the entire state space, it is difficult to determine whether the trajectories of the oscillators in this network are confined in the contraction region of the virtual system. Note that even if a trajectory starts in a contraction region, it may leave the region later. One way to determine whether a trajectory will remain within a contraction region based on the location of the trajectory's starting point is to find a \textit{trapping region} that is inside the contraction region. A trapping region \(N\) is a compact subset of the state space such that the flow of the system is inward everywhere on the boundary of \(N\), and therefore every trajectory that starts within \(N\) will remain there for all future time, details of which can be found in Section 4.

\subsection{Ring Topology}
We cannot generally find virtual systems of bidirectional ring networks in which each oscillator affects its two neighboring oscillators and unidirectional ring networks in which each oscillator affects only its next oscillator. Instead, we introduce a multivariate ``virtual network" as follows. For a network (\ref{2}), we attempt to construct 

\begin{equation*}
	\dot{\pmb{y}_i}=\pmb{\Phi }\left(\pmb{y}_i,\pmb{y}_{i+1},\text{...},\pmb{y}_{i+N-1},\pmb{x}_1,\pmb{x}_2,\text{...},\pmb{x}_N\right), i =1,\cdots
	,N,
\end{equation*}

\noindent for virtual variables \(\pmb{y}_1,\text{...},\pmb{y}_N\), which satisfies 

\begin{equation}\label{20}
	\pmb{\Phi}\left(\pmb{x}_i,\pmb{x}_{i+1},\text{...},\pmb{x}_{i+N-1},\pmb{x}_1,\pmb{x}_2,\text{...},\pmb{x}_N\right)=\pmb{f}\left(\pmb{x}_i\right)+\pmb{u}_i, \forall i =1,\cdots
	,N.
\end{equation}

\noindent where all subscripts are calculated modulo \(N\). We can turn this virtual network into a single \(nN\)-dimensional virtual system of \(\pmb{Y}\) by ``concatenating" all original \(n\)-dimensional virtual state vectors \(\pmb{y}_i\) to form an \(nN\)-dimensional virtual state variable \(\pmb{Y}\), that is

\begin{equation*}
	\pmb{Y}=\left(\pmb{y}_{1}\text{  }\pmb{y}_{2}\text{  }\cdots \text{  }\pmb{y}_{N}\right){}^T\in \mathbb{R}^{nN}.
\end{equation*}

\noindent Then we have

\begin{equation}\label{21}
	\dot{\pmb{Y}}=\pmb{\Psi}\left(\pmb{Y},\pmb{x}_1,\pmb{x}_2,\text{...},\pmb{x}_N\right)\equiv\begin{pmatrix} 
		\pmb{\Phi} \left(\pmb{y}_1,\pmb{y}_2,\text{...},\pmb{y}_N,\pmb{x}_1,\pmb{x}_2,\text{...},\pmb{x}_N\right)\\ 
		\pmb{\Phi} \left(\pmb{y}_2,\text{...},\pmb{y}_N,\pmb{y}_1,\pmb{x}_1,\pmb{x}_2,\text{...},\pmb{x}_N\right) \\ 
		\vdots \\
		\pmb{\Phi} \left(\pmb{y}_N,\pmb{y}_1,\text{...},\pmb{y}_{N-1},\pmb{x}_1,\pmb{x}_2,\text{...},\pmb{x}_N\right)
	\end{pmatrix}
\end{equation}

\begin{theorem}\label{thm3.2}
	If network (\ref{2}) has a virtual system (\ref{21}) which is contracting with respect to \(\pmb{Y}\), then the network is completely synchronized regardless of the initial conditions.
\end{theorem}

\begin{proof}
	Since the function \(\pmb{\Phi}\) satisfies Eq.~(\ref{20}), the trajectory \(\left(\pmb{x}_{1}\text{  }\pmb{x}_{2}\text{  }\cdots \text{  }\pmb{x}_{N}\right){}^T\) formed by concatenating the trajectories of all the oscillators in the original network (\ref{2}) and all its cyclic permutations are particular solutions of the system (\ref{21}). 
	In particular, \(\left(\pmb{x}_{1}\text{  }\pmb{x}_{2}\text{  }\cdots \text{  }\pmb{x}_{N}\right){}^T\) and \(\left(\pmb{x}_{2}\text{  }\cdots \text{  }\pmb{x}_{N}\text{  }\pmb{x}_{1}\right){}^T\) are two solutions of (\ref{21}). Since the system is contracting, they converge exponentially to each other, which implies that \(\pmb{x}_{1}\sim\pmb{x}_{2}, \text{  }\cdots,\text{  }\pmb{x}_{N-1}\sim\pmb{x}_{N}, \text{  and }\pmb{x}_{N}\sim\pmb{x}_{1}\), where \(\pmb{a}\sim \pmb{b}\) denotes that \(\pmb{a}\) and \(\pmb{b}\) are converged. Therefore, all oscillators converge into the same synchronous state vector.
\end{proof} 

We now consider ring networks consisting of \(N\) identical coupled oscillators whose links are either unidirectional with identical linear coupling or bidirectional with different linear couplings. Such networks can be expressed as the following two equations,

\begin{subequations}
	\begin{equation}\label{22a*}	
		\dot{\pmb{x}}_i=\pmb{f}\left(\pmb{x}_i\right)+A(\pmb{x}_{i-1}-\pmb{x}_i), \text{  }\forall i =1,\cdots
		,N.
	\end{equation}
	\begin{equation}\label{22b*}	
		\dot{\pmb{x}}_i=\pmb{f}\left(\pmb{x}_i\right)+A_{i+1,i}(\pmb{x}_{i+1}-\pmb{x}_{i})+A_{i-1,i}(\pmb{x}_{i-1}-\pmb{x}_{i}), \text{  }\forall i =1,\cdots
		,N,
	\end{equation}
\end{subequations}

\noindent where \(A_{i,j}\) denotes coupling from the \(i^{th}\) to \(j^{th}\) node, and the subscripts \(i,j\) are calculated modulo \(N\).

Wang and Slotine discussed the conditions for synchronization of an unidirectional ring network when \(A\) is a symmetric matrix and of a bidirectional ring network when the couplings are \textit{interactional}, i.e., \(A_{i,j}=A_{j,i}\)\cite{wang}. In the following, we relax both of these conditions. We need the following lemma to prove Theorems \ref{thm3.3} and \ref{thm3.4}.

\begin{lemma}\label{lem3.1}
	Let \(M=\left(\begin{array}{cc}
		A &B \\
		B^{T} & D \\
	\end{array}\right)\) be a positive semi-definite \(2\times2\) block matrix. Define \(M^{m,n}\) to be a block matrix of the form
	\begin{equation*}
		M^{m,n}=\left(\begin{array}{ccccc}
			&\vdots& &\vdots& \\
			\cdots&(A)_{mm}&\cdots&(B)_{mn}&\cdots  \\
			&\vdots& &\vdots& \\
			\cdots&(B^{T})_{nm}&\cdots&(D)_{nn}&\cdots  \\
			&\vdots& &\vdots& \\
		\end{array}\right).
	\end{equation*}	
	\noindent where \((H)_{ij}\) denotes that the block \(H\) is in the \(i^{th}\) ``row"  and \(j^{th}\) ``column" of the matrix  \(M^{m,n}\). Here, all the blocks in \(M^{m,n}\) are zero matrices, except for the four blocks written out. Then for any unequal \(m\) and \(n\), \(M \succeq 0\) iff \(M^{m,n} \succeq 0\).
\end{lemma}
One can be easily proved this lemma by noticing that \(\pmb{x}^TM\pmb{x}=\pmb{y}^TM^{m,n}\pmb{y}\), where \(\pmb{x}\) is an arbitrary vector which has dimensions that are compatible with matrix \(M\), and \(\pmb{x}\) is exported to vector \(\pmb{y}\) by arbitrarily choosing the missing elements to have appropriate dimensions according to \(M^{m,n}\). 

\begin{theorem}\label{thm3.3}
	The unidirectional ring network (\ref{22a*}) is completely synhcronized if \({J_{f}}_s+A_{s}\prec 0\) on \(D\), \(A_{s} \succ 0\), and \(4A_{s}-A A_{s}^{-1} A^{T} \succeq 0\).
\end{theorem}

\begin{proof}
	The virtual network of (\ref{22a*}) can be constructed as
	\begin{equation}\label{23*}
		\dot{\pmb{y}}_i=\pmb{f}\left(\pmb{y}_i\right)+A(\pmb{y}_{i-1}-\pmb{y}_{i}), \text{  } i =1,\cdots
		,N.
	\end{equation}
	It is easy to check Eq.~(\ref{23*}) by replacing \(\pmb{y}_j\) with \(\pmb{x}_j\) for all \(j\). Write (\ref{23*}) in the form of (\ref{21}) by concatenating \(\pmb{y}_{i}\). 
	The symmetric part of the Jacobian of \(\pmb{\Psi}\) is
	\begin{widetext}
		\begin{small}
			\begin{equation*}
				\begin{split}
					(\frac{\partial \pmb{\Psi}}{\partial {\pmb{Y}}})_{s}=&\left(\begin{array}{cccc}
						{J_{\pmb{f}1}}_s & & & \\
						& {J_{\pmb{f}2}}_s & & \\
						& & \ddots & \\
						& &  & {J_{\pmb{f}N}}_s
					\end{array}\right)
					-\left(\begin{array}{cccc}
						A_{s} & -\frac{A^{T}}{2} & & -\frac{A}{2} \\
						-\frac{A}{2}  & A_{s} & \ddots &  \\
						& \ddots & \ddots & -\frac{A^{T}}{2}\\
						-\frac{A^{T}}{2}  &  & -\frac{A}{2} & A_{s}
					\end{array}\right)
					=\left(\begin{array}{cccc}
						{J_{\pmb{f}1}}_s+A_{s} &  &  & \\
						& {J_{\pmb{f}2}}_s+A_{s} &  &  \\
						&  & \ddots & \\
						&  &  & {J_{\pmb{f}N}}_s+A_{s}
					\end{array}\right)
					-\left(\begin{array}{cccc}
						2A_{s} & -\frac{A^{T}}{2} & & -\frac{A}{2} \\
						-\frac{A}{2}  & 2A_{s} & \ddots &  \\
						& \ddots & \ddots & -\frac{A^{T}}{2}\\
						-\frac{A^{T}}{2}  &  & -\frac{A}{2} & 2A_{s}
					\end{array}\right),
				\end{split}
			\end{equation*}
		\end{small}
	\end{widetext}
	
	\noindent where \(J_{\pmb{f}i}=\frac{\partial \pmb{f}(\pmb{y}_i)}{\partial {\pmb{y}_i}}\). By Theorem \ref{thm3.2}, we need to prove that \((\frac{\partial \pmb{\Psi}}{\partial {\pmb{Y}}})_{s}\) is uniformly negative definite under the given conditions. For an arbitrary \(\pmb{Y}\), since \({J_{\pmb{f}}}_s+A_{s}\prec 0\) on \(D\), we have \({J_{\pmb{f}i}}_{s}+A_{s} \prec 0\) for all \(i\), and thus the first part in \((\frac{\partial \pmb{\Psi}}{\partial {\pmb{Y}}})_{s}\) is uniformly negative definite.
	
	Define block matrices:
	
	\begin{equation*}
		\begin{split}
			\tilde{A}_1=&\left(\begin{array}{cc}
				A_{s} & -\frac{A^{T}}{2} \\
				-\frac{A}{2} & A_{s}  \\
			\end{array}\right), \\
			\tilde{A}_2=&\left(\begin{array}{cc}
				A_{s} & -\frac{A}{2} \\
				-\frac{A^{T}}{2} & A_{s}  \\
			\end{array}\right).
		\end{split}
	\end{equation*}
	
	Observe that the second part in \((\frac{\partial \pmb{\Psi}}{\partial {\pmb{Y}}})_{s}\) can be written as $\sum_{i=1}^{N-1} {\tilde{A}_1}^{i,i+1}+{\tilde{A}_2}^{1,N}$. Since we have assumed that \(A_{s} \succ 0\) and \(4A_{s}-A A_{s}^{-1} A^{T} \succeq 0\), Schur complement lemma implies that both \(\tilde{A}_1\) and \(\tilde{A}_2\) are positive semi-definite\cite{gallier}. Then, using Lemma \ref{lem3.1}, we have \({\tilde{A}_1}^{i,i+1} \succeq 0\) for \(i\) from 1 to \(N-1\) and \({\tilde{A}_2}^{1,N} \succeq 0\), and thus the second part in \((\frac{\partial \pmb{\Psi}}{\partial {\pmb{Y}}})_{s}\) is positive semi-definite. Therefore, \((\frac{\partial \pmb{\Psi}}{\partial {\pmb{Y}}})_{s}\) is negative definite, and it follows that (\ref{22a*}) is completely synhcronized.
\end{proof}

\begin{theorem}\label{thm3.4}
	The bidirectional ring network (\ref{22b*}) is completely synhcronized if \({J_{\pmb{f}}}_s \prec 0\) and for each \(i\) from 1 to \(N\) at least one of the following two conditions is satisfied
	
	1. \({A_{i+1,i}}_s \succ 0\) and \({A_{i,i+1}}_s-\frac{1}{4}(A_{i,i+1}+A_{i+1,i}^T){A^{-1}_{i+1,i}}_s(A_{i+1,i}+A_{i,i+1}^T) \succeq 0\);
	
	2. \({A_{i,i+1}}_s \succ 0\) and \({A_{i+1,i}}_s-\frac{1}{4}(A_{i+1,i}+A_{i,i+1}^T){A^{-1}_{i,i+1}}_s(A_{i,i+1}+A_{i+1,i}^T) \succeq 0\).
\end{theorem}

\begin{proof}
	The virtual network of (\ref{22b*}) can be constructed as
	\begin{equation}\label{24*}
		\dot{\pmb{y}}_i=\pmb{f}\left(\pmb{y}_i\right)+A_{i+1,i}(\pmb{y}_{i+1}-\pmb{y}_{i})+A_{i-1,i}(\pmb{y}_{i-1}-\pmb{y}_{i}), \text{  } i =1,\cdots,N.
	\end{equation}
	
	Define block matrices:
	
	\begin{equation*}
		\tilde{A}_{j,i}=\left(\begin{array}{cc}
			{A_{j,i}}_{s} & -\frac{(A_{j,i}+A_{i,j}^T) }{2}\\
			-\frac{(A_{i,j}+A_{j,i}^T)}{2} & {A_{i,j}}_{s}  \\
		\end{array}\right), 
	\end{equation*}
	
	Writing the virtual network (\ref{24*}) in the form of virtual system (\ref{21}), the symmetric part of \(\frac{\partial \pmb{\Psi}}{\partial {\pmb{Y}}}\) is \(D_{{J_{\pmb{f}}}_s}-\sum_{i=1}^{N-1} {\tilde{A}_{i+1,i}}^{i,i+1}-{\tilde{A}_{N,1}}^{1,N}\), where \(D_{{J_{\pmb{f}}}_s}\) is a block diagonal matrix whoes the main-diagonal blocks are \({J_{\pmb{f}i}}_{s}\). The fact that \({J_{\pmb{f}}}_s\) is negative definite guarantees that \(D_{{J_{\pmb{f}}}_s} \prec 0\). By Schur complement lemma, for each \(i\) from \(1\) to \(N\), when \({A_{i+1,i}}_s\) is invertable, \(\tilde{A}_{i+1,i} \succeq 0\) and \(\tilde{A}_{i,i+1} \succeq 0\) iff \({A_{i+1,i}}_s \succ 0\) and \({A_{i,i+1}}_s-\frac{1}{4}(A_{i,i+1}+A_{i+1,i}^T){A^{-1}_{i+1,i}}_s(A_{i+1,i}+A_{i,i+1}^T) \succeq 0\); while when \({A_{i,i+1}}_s\) is invertable, \(\tilde{A}_{i+1,i} \succeq 0\) and \(\tilde{A}_{i,i+1} \succeq 0\) iff \({A_{i,i+1}}_s \succ 0\) and \({A_{i+1,i}}_s-\frac{1}{4}(A_{i+1,i}+A_{i,i+1}^T){A^{-1}_{i,i+1}}_s(A_{i,i+1}+A_{i+1,i}^T) \succeq 0\). Thus, we have \(\tilde{A}_{i+1,i} \succeq 0\) for \(i\) from \(1\) to \(N-1\) and \(\tilde{A}_{N,1} \succeq 0\). As a consequence of Lemma \ref{lem3.1}, \({\tilde{A}_{i+1,i}}^{i,i+1} \succeq 0\) for \(i\) from \(1\) to \(N-1\) and \({\tilde{A}_{N,1}}^{1,N} \succeq 0\). Then \((\frac{\partial \pmb{\Psi}}{\partial {\pmb{Y}}})_{s} \prec 0\). Hence (\ref{22b*}) is completely synchronized.
\end{proof}

The following two corollaries can be drawn from Theorems \ref{thm3.3} and \ref{thm3.4}:

\begin{corollary}\label{cor3.1}
	The unidirectional ring network (\ref{22a*}) with identical symmetric coupling matrix \(A\) is completely synhcronized if \({J_{\pmb{f}}}_s+A\prec 0\) on \(D\) and \(A \succ 0\);
\end{corollary}

\begin{corollary}\label{cor3.2}
	The bidirectional ring network (\ref{22b*}) with interactional couplings, \(A_{i,j}=A_{j,i}\), is completely synhcronized if \({J_{\pmb{f}}}_s \prec 0\) and \({A_{i,i+1}}_s \succ 0\) for all \(i\) from 1 to \(N\).
\end{corollary}

We've revovered the results obtained in Ref. \onlinecite{wang}.

\section{Networks of the (Duffing/Rayleigh) van der Pol Systems and Numerical Simulations}
\subsection{Van der Pol Oscillators}
We first consider a network of \(N\) all-to-all coupled identical vdP oscillators with symmetric linear couplings. In the absence of coupling, a single vdP oscillator has the following equation of motion

\begin{equation} \label{31*}
	\ddot{x}+\left(\alpha x^{2}-\gamma\right) \dot{x}+\omega^{2}x=0.
\end{equation}

It is usually written in two different two-dimensional forms, one by setting \(y=\dot{x}\) and the other by setting \(y=\frac{1}{\omega}(\dot{x}+\frac{\alpha}{3}x^3-\gamma x)\). We use the latter because, in this case, the Jacobian \(J_{\pmb{f}}\) has a globally negative semi-definite symmetric part if \(\alpha \geqslant 0\) and \(\gamma \leqslant 0\). Consider a network of \(N\) all-to-all coupled identical vdP oscillators with symmetric linear couplings described by Eq.~(\ref{5*}). Write 
\begin{equation} \label{sum}
	\sum_{j=1}^{N} {A_j}_s = \left(\begin{array}{cc}
		A^{11}_{ts} & A^{12}_{ts} \\
		A^{21}_{ts} & A^{22}_{ts} \\
	\end{array}\right)
\end{equation}
with \(A^{12}_{ts}=A^{21}_{ts}\). 
We can then establish the following theorem.

\begin{theorem}\label{thm4.1}
	For a network of \(N\) all-to-all coupled identical vdP oscillators with symmetric linear couplings described by Eq.~(\ref{5*}) in which \(\pmb{f}\) is defined by
	\begin{equation} \label{32*}
		\pmb{f}(x,y)=\left(\begin{array}{c}\omega y-\frac{\alpha}{3}x^3+\gamma x 	\\-\omega x\end{array}\right),
	\end{equation}
	if \(\alpha \geqslant 0\), \(\gamma < A^{11}_{ts}-\frac{(A^{12}_{ts})^2}{A^{22}_{ts}}\), and \(A^{22}_{ts} > 0\), then the network is completely synchronized.
\end{theorem}

\begin{proof}
	According to Proposition \ref{prop3.1}, the network is completely synchronized if
	
	\begin{equation*}
		{J_{\pmb{f}}}_s-\sum_{j=1}^{N} {A_j}_s= \left(\begin{array}{cc}
			\gamma-\alpha x^2-A^{11}_{ts} & -A^{12}_{ts} \\
			-A^{21}_{ts} & -A^{22}_{ts} \\
		\end{array}\right)
	\end{equation*}
	
	\noindent is uniformly negative definite on the entire state space \(\mathfrak{D}\), or, equivalently, \(-({J_{\pmb{f}}}_s-\sum_{j=1}^{N} {A_j}_s) \succ 0\) on \(\mathfrak{D}\). Since \(\gamma < A^{11}_{ts}-\frac{(A^{12}_{ts})^2}{A^{22}_{ts}}\), \(\alpha \geqslant 0\) and \(A^{22}_{ts} > 0\), we have
	
	\begin{equation*}
		-(\gamma-\alpha x^2-A^{11}_{ts}) \geqslant -(\gamma-A^{11}_{ts}) \geqslant -(\gamma-A^{11}_{ts}+\frac{(A^{12}_{ts})^2}{A^{22}_{ts}}) > 0,
	\end{equation*}
	
	\noindent and
	
	\begin{equation*}
		-(\gamma-A^{11}_{ts})A^{22}_{ts}-(A^{12}_{ts})^2 > 0.
	\end{equation*}
	
	\noindent Consequently,
	
	\begin{equation*}
		- (-\gamma+\alpha x^2+A^{11}_{ts})^{-1} \geqslant - (-\gamma+A^{11}_{ts})^{-1},
	\end{equation*}
	
	\noindent and
	
	\begin{equation*}
		\begin{split}
			0< &A^{22}_{ts}- (-\gamma+A^{11}_{ts})^{-1} (A^{12}_{ts})^2 \\
			\leqslant &A^{22}_{ts}- (-\gamma+\alpha x^2+A^{11}_{ts})^{-1} (A^{12}_{ts})^2 \\
			= &A^{22}_{ts}- A^{21}_{ts}(-\gamma+\alpha x^2+A^{11}_{ts})^{-1} A^{12}_{ts}.
		\end{split}
	\end{equation*}
	
	\noindent Schur complement lemma then implies \(-({J_{\pmb{f}}}_s-\sum_{j=1}^{N} {A_j}_s) \succ 0\).
\end{proof}

Next, we consider \(N\) identical vdP oscillators coupled by a star network with symmetric linear couplings described by Eq.~(\ref{8}) where the sum of coupling matrices is written as Eq.~(\ref{sum}) and \({A_1}_s\) is given by 

\begin{equation*}
	{A_1}_s = \left(\begin{array}{cc}
		A^{11}_{1s} & A^{12}_{1s} \\
		A^{21}_{1s} & A^{22}_{1s} \\
	\end{array}\right)
\end{equation*}
with \(A^{12}_{1s}=A^{21}_{1s}\). Then using Proposition \ref{prop3.2} and following similar arguments as in the proof of Theorem \ref{thm4.1}, we have

\begin{theorem}\label{thm4.2}
	For \(N\) identical vdP oscillators coupled by a star network with symmetric linear couplings described by Eq.~(\ref{8}) in which \(\pmb{f}\) is given by Eq.~(\ref{32*}), if \(\alpha \geqslant 0\), \(\gamma < \text{min} \{A^{11}_{1s}-\frac{(A^{12}_{1s})^2}{A^{22}_{1s}},A^{11}_{ts}-\frac{(A^{12}_{ts})^2}{A^{22}_{ts}}\}\), and \(\text{min} \{A^{22}_{1s},A^{22}_{ts}\} > 0\), then the network is completely synchronized.
\end{theorem}

\begin{example}
	Consider a star network of six coupled identical vdP oscillators described by Eqs.~(\ref{8}) and (\ref{32*}) with \(\alpha = 1\), \(\omega = 1\), and \(\gamma  = 5\), where the coupling matrices are given by
	
	\begin{align*}
		& A_1=\left(\begin{array}{cc}
			8 & 1 \\
			3 & 4 \\
		\end{array}\right), \;
		A_2=\left(\begin{array}{cc}
			3 & 3 \\
			4 & -5 \\
		\end{array}\right), \;
		A_3=\left(\begin{array}{cc}
			7 & -2 \\
			-5 & -2 \\
		\end{array}\right), \\
		& A_4=\left(\begin{array}{cc}
			1 & 2 \\
			4 & 2 \\
		\end{array}\right), \;
		A_5=\left(\begin{array}{cc}
			8 & -3 \\
			1 & 10 \\
		\end{array}\right), \;
		A_6=\left(\begin{array}{cc}
			-3 & -4 \\
			2 & -6 \\
		\end{array}\right).
	\end{align*} 
	
	\noindent In this case, \( \text{min} \{A^{22}_{1s},A^{22}_{ts}\}=\text{min}\{4,3\}=3>0\), and
	
	\begin{equation*}
		5=\gamma < \text{min} \{A^{11}_{1s}-\frac{(A^{12}_{1s})^2}{A^{22}_{1s}},A^{11}_{ts}-\frac{(A^{12}_{ts})^2}{A^{22}_{ts}}\} =\text{min}\{7,21\}=7.
	\end{equation*}
	
	\noindent By Theorem \ref{thm4.2}, the network is completely synchronized.
	
	Figure \ref{fig4.1} shows the complete synchronization phenomenon of this network. It can be seen that all six oscillators converge into synchronous \(x\) and \(y\) states very quickly, as expected.
	
	\begin{figure}[htbp]
	\centering
	\begin{subfigure}{.25\textwidth}
		\centering
		\includegraphics[width=.8\linewidth]{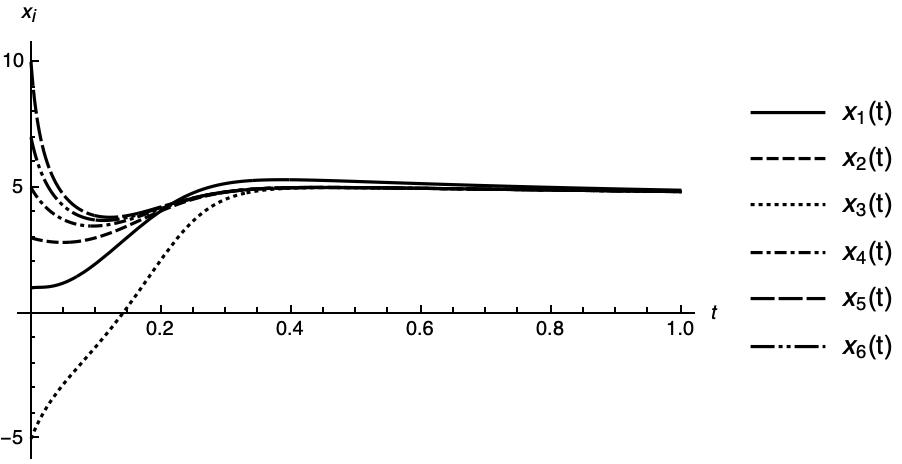}
		\caption{The time series of \(x_i\) for \(i\) from 1 to 6}
		\label{fig4.1:sub1}
	\end{subfigure}%
	\begin{subfigure}{.25\textwidth}
		\centering
		\includegraphics[width=.8\linewidth]{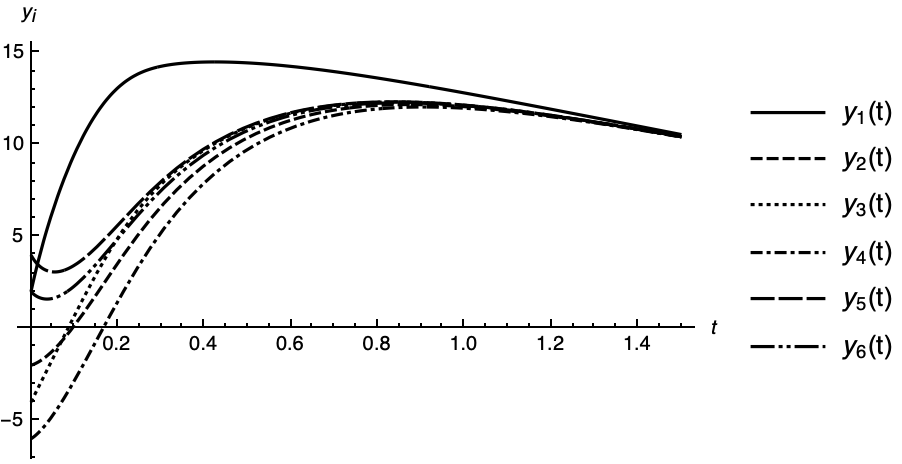}
		\caption{The time series of \(y_i\) for \(i\) from 1 to 6}
		\label{fig4.1:sub2}
	\end{subfigure}
	\caption{The time series of self-sustained six identical vdP oscillators coupled by a star network with symmetric linear couplings}
	\label{fig4.1}
\end{figure}

\end{example}

\subsection{Duffing van der Pol Oscillators}
We now introduce a Duffing term, \(x^3\), into the vdP equation (\ref{31*}) so that the euqation becomes

\begin{equation} \label{33*}
	\ddot{x}+\left(\alpha x^{2}-\gamma\right) \dot{x}+(\omega^{2}+\varepsilon x^{2})x=0.
\end{equation}

\noindent This is the equation of motion of a DvdP oscillator. Suppose \(\alpha \geqslant 0\), \(-(\omega^2+1)<\gamma<-1\), \(\omega \neq 0\), and \(\varepsilon=\frac{\alpha}{3}\). Let \(\gamma^{\prime}=\gamma+1\). Since \(-(\omega^2+1)<\gamma<-1\), we have \(-\omega^2<\gamma^{\prime}<0\). We can therefore let \({\omega^{\prime}}^2=\omega^2+\gamma^{\prime}>0\). Equation (\ref{33*}) can thus be written as

\begin{gather*}
	\ddot{x}+\left(\alpha x^{2}-\gamma^{\prime}+1\right) \dot{x}+({\omega^{\prime}}^2-\gamma^{\prime}+\frac{\alpha}{3} x^{2})x=0 \\
	\frac{d}{dt}(\dot{x}+\frac{\alpha}{3} x^3-\gamma^{\prime} x)=\omega^{\prime}[-\omega^{\prime} x-\frac{1}{\omega^{\prime}} (\dot{x}+\frac{\alpha}{3} x^3-\gamma^{\prime} x)]
\end{gather*}

Making a transformation \(y=\frac{1}{\omega^{\prime}}(\dot{x}+\frac{\alpha}{3} x^3-\gamma^{\prime} x)\), we get a two dimensional system

\begin{equation}\label{34*}
	\left(\begin{array}{l}\dot{x} \\\dot{y}\end{array}\right)=\pmb{f}(x,y)=\left(\begin{array}{c}\omega^{\prime} y-\frac{\alpha}{3}x^3+\gamma^{\prime} x 	\\-\omega^{\prime} x-y\end{array}\right).
\end{equation}

\noindent with

\begin{equation*}
	{J_{\pmb{f}}}_s=\left(\begin{array}{cc}
		\gamma^{\prime}-\alpha x^2 &0 \\
		0 & -1 \\
	\end{array}\right).
\end{equation*}

\noindent It is obvious that \({J_{\pmb{f}}}_s \prec 0\) as \(\alpha \geqslant 0\) and  \(\gamma^{\prime}<0\). Using Corollary \ref{cor3.2}, we conclude the following theorem

\begin{theorem}\label{thm4.3}
	For a bidirectional ring network of \(N\) identical coupled DvdP oscillators with interactional linear couplings described by Eq.~(\ref{22b*}) in which \(\pmb{f}\) is given by Eq.~(\ref{34*}), if \(\alpha \geqslant 0\), \(\gamma^{\prime}<0\), and \({A_{i,i+1}}_{s}\succ 0\) for all \(i\), then the network is completely synchronized.
\end{theorem}

\subsection{Rayleigh van der Pol Oscillators}
Next, we introduce a Rayleigh term, \(\dot{x}^3\), into the vdP equation (\ref{31*}) to get a RvdP oscillator
\begin{equation} \label{35*}
	\ddot{x}+\left(\alpha x^{2}+\beta \dot{x}^2-\gamma\right) \dot{x}+\omega^{2}x=0.
\end{equation}

Letting \(y=\dot{x}\), its two-dimensional form can be written as

\begin{equation}\label{36*}
	\left(\begin{array}{l}\dot{x} \\\dot{y}\end{array}\right)=\pmb{f}(x,y)=\left(\begin{array}{c} y \\-\left(\alpha x^{2}+\beta y^2-\gamma\right) y-\omega^{2}x\end{array}\right).
\end{equation}

We assume all-to-all symmetric coupling as described in Eq.~(\ref{5*}). We also assume \textit{full-state} linear coupling, i.e., the coupling matrices \(A_{i}\) are diagonal for all \(i\). If \(x_{i}\) and \(y_{i}\) represent the ``position" and ``velocity" of \(i^{\text{th}}\) node, respectively, then full-state coupling means that each node adjusts its position and velocity according to the difference between its position and velocity with other nodes, respectively\cite{alderisio}. We can write

\begin{equation}\label{37*}
	{J_{\pmb{f}}}_s-\sum_{j=1}^{N} {A_{j}}_s=\left(\begin{array}{cc}
		-\frac{c}{N}\sum_{j=1}^{N} a_{j} &\frac{1-2\alpha xy-\omega^2}{2} \\
		\frac{1-2\alpha xy-\omega^2}{2} & -(\alpha x^2+3\beta y^2-\gamma)-\frac{c}{N}\sum_{j=1}^{N} b_{j} \\
	\end{array}\right).
\end{equation}

\noindent where \(c\) is a constant coupling strength, and \(a_{j}\) and \(b_{j}\) are coefficients that represent the coupling effect of \(j^{\text{th}}\) node on others. In this case, \({J_{\pmb{f}}}_s\) is neither negative definite nor negative semi-definite. Furthermore, the definiteness of \({J_{\pmb{f}}}_s-\sum_{j=1}^{N} {A_{j}}_s\) cannot be determined globally since it depends on the position \((x,y)\) in the virtual phase plane. We can nevertheless find the contraction region of the virtual system of this network:

\begin{lemma}\label{lem4.1}
	For a network of \(N\) all-to-all coupled identical RvdP oscillators with symmetric and full-state linear couplings described by Eq.~(\ref{5*}) in which \(\pmb{f}\) is given by Eq.~(\ref{36*}), if \(\frac{c}{N}\sum_{j=1}^{N} a_{j}>0\), then the network is completely synchronized if all trajectories, \((x_{k}(t)\text{  }y_{k}(t))^T\), start within the region \(\mathfrak{C}_{RvdP}:=\left\{(x\text{    }y)^T\in \mathbb{R}^2: \frac{c^2}{N^2}\sum_{i, j=1}^{N} a_{i} b_{j}+\mu(x,y) \frac{c}{N}\sum_{j=1}^{N} a_{j}-\nu^2(x,y)>0\right\}\) and ramain there for \(k=1, \cdots, N\), where \(\mu(x,y)\equiv\alpha  x^2+3 \beta  y^2-\gamma\) and \(\nu(x,y)\equiv\frac{1-2 \alpha  x y-\omega^2}{2}\).
\end{lemma}

\begin{proof}
	See Appendix \ref{appB:subsec1}.
\end{proof}

To determine whether all trajectories will remain within \(\mathfrak{C}_{RvdP}\), we need to find a trapping region of the virtual system that is contained in \(\mathfrak{C}_{RvdP}\). For simplicity, we make a further assumption that the coupling matrices are identical and symmetric, i.e., \(a_{i}=b_{i}=a\) for all \(i\). The dynamics of this network can be expressed as

\begin{equation}\label{38*}
	\begin{split}
		\dot{x}_{i}&=y_{i}+\frac{c}{N}a\sum_{j=1}^{N} (x_{j}-x_{i}) \\
		\dot{y}_{i}&=-\left(\alpha {x_{i}}^{2}+\beta {y_{i}}^2-\gamma\right) y_{i}-\omega^{2}x_{i}+\frac{c}{N}a\sum_{j=1}^{N} (y_{j}-y_{i}).
	\end{split}
\end{equation}

Using the LaSalle's invariance principle, we can obtain the following lemma.

\begin{lemma}\label{lem4.2}
	For a network of \(N\) all-to-all coupled identical RvdP oscillators with identical full-state coupling described by Eq.~(\ref{38*}), the origin is globally asymptotically stable if \(\alpha\geqslant0\), \(\beta\geqslant0\), \(ca \geqslant 0\), \(\omega \neq 0\), and \(\gamma<0\).
\end{lemma}

\begin{proof}
	See Appendix \ref{appB:subsec2}.
\end{proof}

As a result, the solution of Eq.~(\ref{38*}) with \(\alpha\geqslant0\), \(\beta\geqslant0\), \(ca \geqslant 0\), \(\omega \neq 0\), and \(\gamma<0\) will eventually converge to the origin regardless of the initial conditions. However, as stated in Definition \ref{def3.1}, usually \(\pmb{0}\) is not treated as a synchronous state. We will show that all overdamped coupled RvdP oscillators can reach complete synchronization before they stop oscillating.

Observe that the virtual system of (\ref{38*}) can be constrcuted as follows:

\begin{equation}\label{39*}
	\left(\begin{array}{l}\dot{x} \\\dot{y}\end{array}\right)=\left(\begin{array}{c} y+\frac{ca}{N}(\sum_{j=1}^{N} x_{j}-Nx) \\-\left(\alpha x^{2}+\beta y^2-\gamma\right) y-\omega^{2}x+\frac{ca}{N}(\sum_{j=1}^{N} y_{j}-Ny)\end{array}\right).
\end{equation}

\begin{lemma}\label{lem4.3}
	The disk of radius \(r\)
	
	\({\mathfrak{R}_{RvdP}}_{\omega^2=1}:=\left\{(x,y) \in \mathbb{R}^2: H(x,y)=x^2+y^2 \leqslant r^2\right\}\) 
	
	\noindent is a trapping region of the virtual system (\ref{39*}) with \(\alpha\geqslant0\), \(\beta\geqslant0\), \(ac \geqslant 0\), \(\omega^2 = 1\), and \(\gamma<0\).
\end{lemma}

\begin{proof}
	See Appendix \ref{appB:subsec3}.
\end{proof}

By Lemma \ref{lem4.1}, when \(\omega^2=1\), the contraction region \(\mathfrak{C}_{RvdP}\) of the virtual system (\ref{39*}) becomes

\begin{equation*}
	\begin{split}
		{\mathfrak{C}_{RvdP}}_{\omega^2=1}\equiv&\left\{(x,y)\in \mathbb{R}^2: c^{2}a^{2}+ca\mu(x,y)-\nu^2(x,y)>0\right\} \\
		=&\left\{(x,y)\in \mathbb{R}^2:g(x,y)>0\right\}.
	\end{split}
\end{equation*}

\noindent where \(g(x,y)\equiv c^{2}a^{2}-ca\gamma+ca\alpha  x^2+3ca \beta  y^2- \alpha^2  x^2 y^2\). 

\begin{lemma}\label{lem4.4}
	Provided \(ca >\max \{\gamma,0\}\), \({\mathfrak{R}_{RvdP}}_{\omega^2=1}\) is contained in \({\mathfrak{C}_{RvdP}}_{\omega^2=1}\) in the following four cases
	
	1. \(3 \beta r^2>\gamma-ca\) and \(\alpha = 0\);
	
	2.  \(r^2 \leqslant \frac{ca\alpha-3ca\beta}{\alpha^2}\) and \(3 \beta r^2>\gamma-ca\) with \(0 \neq \alpha>3\beta\);
	
	3. \(r^2 \leqslant \frac{3ca \beta-ca \alpha}{\alpha^2}\) and \(\alpha r^2>\gamma-ca\) with \(0 \neq \alpha<3\beta\);
	
	4.  \(r^2 \geqslant \max\{\frac{ca \alpha-3ca \beta}{\alpha^2}, \frac{3ca \beta-ca \alpha}{\alpha^2}\}\) and
	\begin{align*}
		\frac{ca \alpha+3ca \beta-2\sqrt{3c^2a^2 \alpha \beta+\alpha^2 ca (ca-\gamma)}}{\alpha^2}< r^2 \\
		< \frac{ca \alpha+3ca \beta+2\sqrt{3c^2a^2 \alpha \beta+\alpha^2 ca (ca-\gamma)}}{\alpha^2}
	\end{align*}

	with \(\alpha\) and \(\beta\) satisfying \(3ca\alpha\beta+\alpha^2 (ca-\gamma) > 0\).
\end{lemma}

\begin{proof}
	See Appendix \ref{appB:subsec4}.
\end{proof}

In these cases, any trajectory of the network (\ref{38*}) starting within \({\mathfrak{R}_{RvdP}}_{\omega^2=1}\) with appropriate coefficients will remain there due to Lemma \ref{lem4.3} and thence in \({\mathfrak{C}_{RvdP}}_{\omega^2=1}\) due to Lemma \ref{lem4.4} . Consequently, Lemma \ref{lem4.1} implies the following theorem.

\begin{theorem}\label{thm4.4}
	Consider a network of \(N\) all-to-all coupled identical RvdP oscillators with an identical full-state linear coupling described by Eq.~(\ref{38*}) with \(\alpha \geqslant 0\), \(\beta \geqslant 0\), \(ca >0\), \(\gamma<0\), and \(\omega^2=1\). In the cases stated in Lemma \ref{lem4.4}, the network is completely synchronized if its starting points, \(({x_{i}}_{0},{y_{i}}_{0})\), are taken in the region \({\mathfrak{R}_{RvdP}}_{\omega^2=1}\) for all \(i\) from 1 to \(N\).
\end{theorem}

Note that this theorem is based on the contraction theory, so the synchronization rate is exponential due to Theorem \ref{thm2.1}. However, as can be seen in Lemma \ref{lem4.6} below, the rate of convergence of the network to the origin is also exponential if the initial conditions are in a neighborhood of the origin. Therefore, if the initial conditions are in \({\mathfrak{R}_{RvdP}}_{\omega^2=1}\) and close to the origin, we usually cannot determine whether the network can reach complete synchronization before it decays to the equilibrium position without numerical methods. In general, for large coupling constants \(c\) and \(a\), the complete synchronization can be achieved first. For initial conditions in \({\mathfrak{R}_{RvdP}}_{\omega^2=1}\) and far from the origin, in most cases the network will get completely synchronized before the oscillations stop. 

\begin{example}
	Consider a network of six all-to-all coupled identical RvdP oscillators with an identical full-state linear coupling described by Eq.~(\ref{38*}) with \(\alpha = 1\), \(\beta = 1\), \(\gamma  = -0.1\), \(\omega = 1\), \(a=1\), \(c=1\), and \(N  = 6\). Suppose the starting points of this system are
	
	\begin{gather*}
		(x_1(0),y_1(0))=(1,1.5), \; (x_2(0),y_2(0))=(0.3,-1.2), \\
		(x_3(0),y_3(0))=(-0.5,-1.4), \; (x_4(0),y_4(0))=(0.5,-0.5), \\
		(x_5(0),y_5(0))=(1.7,0.9), \; (x_6(0),y_6(0))=(0.3,-1).
	\end{gather*}
	
	Choose the trapping region, \({\mathfrak{R}_{RvdP}}_{\omega^2=1}\), to be a disk of radius \(r=2\). Then all the starting points are in \({\mathfrak{R}_{RvdP}}_{\omega^2=1}\). By Lemma \ref{lem4.4} (case 4 of the lemma), \({\mathfrak{R}_{RvdP}}_{\omega^2=1}\) is included in the contraction region, \({\mathfrak{C}_{RvdP}}_{\omega^2=1}\), of the virtual system of this network since
	
	\begin{gather*}
		4=r^2 \geqslant \max\{\frac{ca \alpha-3ca \beta}{\alpha^2}, \frac{3ca \beta-ca \alpha}{\alpha^2}\}=2, \\
		-0.05 \approx \frac{ca \alpha+3ca \beta-2\sqrt{3c^2a^2 \alpha \beta+\alpha^2 ca (ca-\gamma)}}{\alpha^2}< r^2 \\
		< \frac{ca \alpha+3ca \beta+2\sqrt{3c^2a^2 \alpha \beta+\alpha^2 ca (ca-\gamma)}}{\alpha^2} \approx 8.05, \\
		3ca\alpha\beta+\alpha^2 (ca-\gamma) =4.1 >0.
	\end{gather*}
	
	\noindent By Theorem \ref{thm4.4}, the network is completely synchronized.
	
	It can be seen from Fig. \ref{fig4.2} that the trajectories, \((x_i(t),y_i(t))\), of all oscillators in this network converge to the origin as predicted by Lemma \ref{lem4.2}, but all oscillators quickly reach complete synchronization before the amplitudes of the oscillations decrease to zero.
	
	\begin{figure}[htbp]
		\centering
		\begin{subfigure}{.25\textwidth}
			\centering
			\includegraphics[width=.8\linewidth]{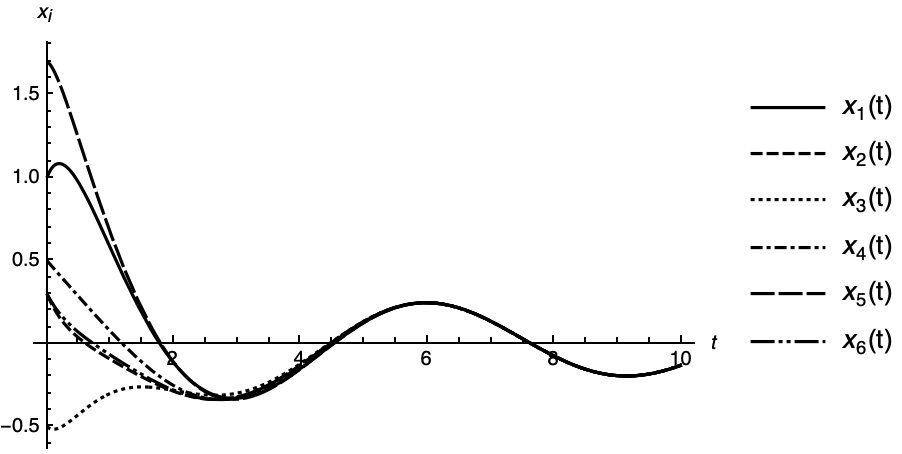}
			\caption{The time series of \(x_i\) for \(i\) from 1 to 6}
			\label{fig4.2:sub1}
		\end{subfigure}%
		\begin{subfigure}{.25\textwidth}
			\centering
			\includegraphics[width=.8\linewidth]{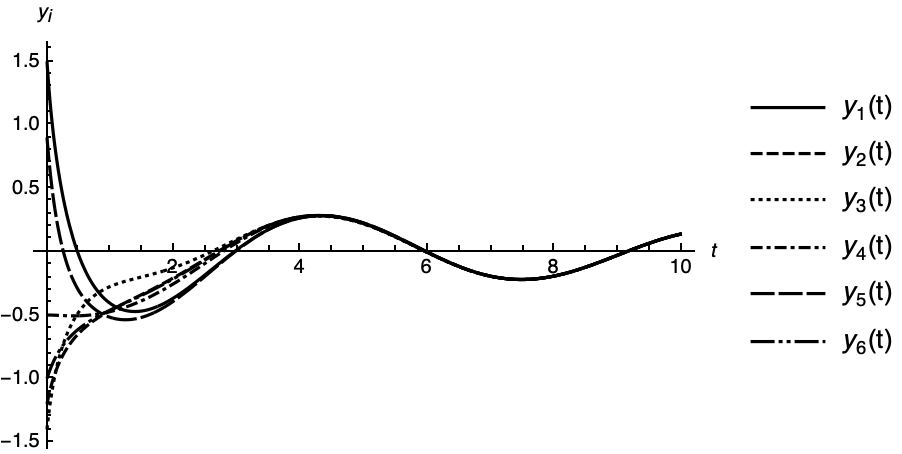}
			\caption{The time series of \(y_i\) for \(i\) from 1 to 6}
			\label{fig4.2:sub2}
		\end{subfigure}
		\begin{subfigure}{.25\textwidth}
			\centering
			\includegraphics[width=.8\linewidth]{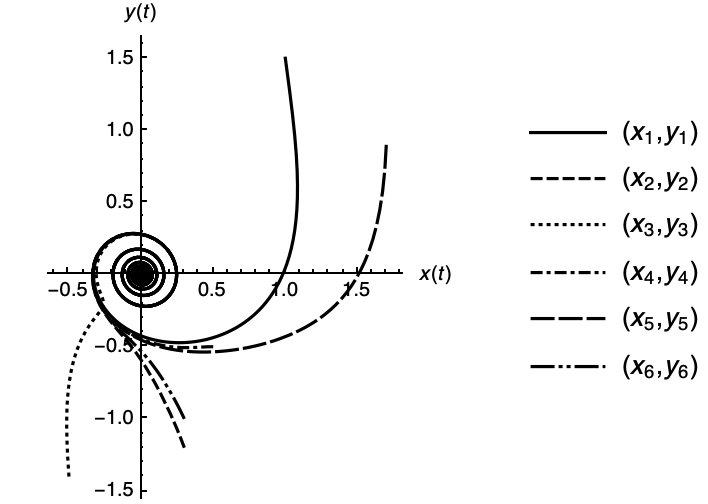}
			\caption{Phase portraits}
			\label{fig4.2:sub3}
		\end{subfigure}
		\caption{The time series and phase portraits of underdamped six all-to-all coupled identical RvdP oscillators with an identical full-state linear coupling}
		\label{fig4.2}
	\end{figure}	
	
\end{example}

Next, we consider the case when \(\gamma>0\). In this case, the origin may not be an asymptotically stable equilibrium point, and the system may exhibit periodic behavior. We attempt to find the conditions that make the system (\ref{38*}) self-sustained so that a stable limit cycle exists in its state space.

The following lemma will be used in the proof of Lemma \ref{lem4.6}.
\begin{lemma}\label{lem4.5}
	For \(N \times N\) matries \(W_N\) and \(S_N\) whose elements are defined as
	\begin{equation*}
		{W_N}_{i j}=\left\{\begin{array}{ll}
			a & \text { if { }} i=j \\
			b & \text{otherwise}
		\end{array}\right.
	\end{equation*}
	
	\noindent and
	
	\begin{equation*}
		{S_N}_{i j}=\left\{\begin{array}{ll}
			a & \text { if { }} i=j \neq 1 \\
			b & \text{otherwise}
		\end{array}\right.
	\end{equation*}
	
	\noindent respectively, we have
	\begin{equation} \label{40*}
		\begin{split}
			\text{det}[W_N]=&(a-b)^{N-1}[a+(N-1)b] \\
			\text{det}[S_N]=&b(a-b)^{N-1}.
		\end{split}
	\end{equation}
\end{lemma}

\begin{proof}
	See Appendix \ref{appB:subsec5}.
\end{proof}

\begin{lemma}\label{lem4.6}
	Suppose \(ac>0\) and \(\omega \neq 0\). When \(|\gamma| \geqslant 2|\omega|\) and \(\gamma<0\), the origin is a locally exponentially stable stationary point for the system (\ref{38*}) (overdamped oscillators for \(|\gamma| >2|\omega|\)). When \(|\gamma| < 2|\omega|\) and \(\gamma<2ca\),  the system (\ref{38*}) experiences a Hopf bifurcation at the origin when \(\gamma\) crosses 0, and the origin is locally exponentially stable for \(\gamma<0\) (underdamped oscillators) and unstable for \(\gamma>0\), and the bifurcation is supercritical if  \(\alpha \geqslant 0\) and \(\beta \geqslant 0\).
	
\end{lemma}

\begin{proof}
	The Jacobian of the linearized system (\ref{38*}) about the origin is given by
	\begin{equation*}
		J=\left(\begin{array}{cc}
			-\frac{ca}{N} L_{N} &I_{N} \\
			-\omega^2 I_{N}  & \gamma I_{N}-\frac{ca}{N} L_{N}  \\
		\end{array}\right),
	\end{equation*}
	\noindent where \(I_{N}\) is an \(N \times N\) identical matrix and \(L_{N}\) is an \(N \times N\) Laplacian matrix of complete graph whose elements are given by
	\begin{equation*}
		{L_N}_{i j}=\left\{\begin{array}{lll}
			N-1 & \text { if { }} i=j \\
			-1 & \text { otherwise }
		\end{array}\right..
	\end{equation*}
	
	\noindent Hence an eigenvalue \(\lambda\) of \(J\) satisfies
	
	\begin{equation*}
		\text{det}\left(\begin{array}{cc}
			-\frac{ca}{N} L_{N}-\lambda I_{N} &I_{N} \\
			-\omega^2 I_{N}  & (\gamma-\lambda) I_{N}-\frac{ca}{N} L_{N}  \\
		\end{array}\right)=0.
	\end{equation*}
	
	\noindent Since \(I_{N}\) is an identical matrix, the lower two blocks in matrix \(J\) commute. Using a property of block matrix determinant, we have
	
	\begin{equation*}
		\begin{split}
			0=&\text{det}\left(\begin{array}{cc}
				-\frac{ca}{N} L_{N}-\lambda I_{N} &I_{N\times N} \\
				-\omega^2 I_{N}  & (\gamma-\lambda) I_{N}-\frac{ca}{N} L_{N}  \\
			\end{array}\right) \\
			=&\text{det}((-\frac{ca}{N} L_{N}-\lambda I_{N})((\gamma-\lambda) I_{N}-\frac{ca}{N} L_{N})+\omega^2 I^2_{N})\\
			=&\text{det}(\frac{c^2a^2}{N^2}L^2_{N}+\frac{ca}{N}(2\lambda-\gamma)L_{N}+[\omega^2-\lambda(\gamma-\lambda)]I_{N}).
		\end{split}
	\end{equation*}
	
	It can be calculated that
	\begin{equation*}
		{L^2_N}_{i j}=\left\{\begin{array}{lll}
			N^2-N & \text { if { }} i=j \\
			-N & \text { otherwise }
		\end{array}\right..
	\end{equation*}
	
	\noindent We therefore have
	\begin{equation*}
		\begin{split}
			0=\text{det}\left(\begin{array}{cccc}
				p(\lambda) & q(\lambda) & \cdots & q(\lambda) \\
				q(\lambda)  & p(\lambda) & \cdots & q(\lambda)  \\
				\vdots & \vdots & \ddots & \vdots\\
				q(\lambda)  & q(\lambda) & \cdots & p(\lambda)
			\end{array}\right)_{N \times N},
		\end{split}
	\end{equation*}
	
	\noindent where 
	\begin{align*}
		p(\lambda)=&\lambda^2+(2ca\frac{N-1}{N}-\gamma)\lambda+c^2a^2-ca\gamma+\frac{ca\gamma-c^2 a^2}{N}+\omega^2, \\
		q(\lambda)=&-\frac{2ca}{N}\lambda+\frac{ca\gamma-c^2a^2}{N}.
	\end{align*}
	
	Using Lemma \ref{lem4.5}, \(\lambda\) satisfies 
	\begin{equation} \label{41*}
		(p-q)^{N-1}[p+(N-1)q](\lambda)=0.
	\end{equation}
	
	\noindent One can find
	
	\begin{equation} \label{42*}
		(p-q)(\lambda)=\lambda^2+(2ca-\gamma)\lambda+c^2a^2-ca\gamma+\omega^2,\\
	\end{equation}
	
	\begin{equation} \label{43*}
		[p+(N-1)q](\lambda)=\lambda^2-\gamma\lambda+\omega^2.
	\end{equation}
	
	\noindent It follows that two roots of the function (\ref{42*}) are a pair of solutions of Eq.~(\ref{41*}) with multiplicity \(N-1\) and two roots of the function (\ref{43*}) are a pair of solutions of Eq.~(\ref{41*}) with multiplicity \(1\). One can find the roots of (\ref{42*}) and (\ref{43*}) are
	
	\begin{align*}
		\lambda_{1\pm}=&\frac{1}{2}(-2ac+\gamma \pm \sqrt{\gamma^2-4\omega^2}), \\
		\lambda_{2\pm}=&\frac{1}{2}(\gamma \pm \sqrt{\gamma^2-4\omega^2})
	\end{align*}
	
	\noindent respectively.  Assume that \(ac \geqslant 0\) and \(\omega \neq 0\). Then for \(\gamma < 0\) and \(|\gamma| \geqslant 2|\omega|\), \(\lambda_{1\pm}\) and \(\lambda_{2\pm}\) are all negative reals since
	\begin{equation*}
		\gamma+\sqrt{\gamma^2-4\omega^2} < 0.
	\end{equation*}
	
	\noindent Therefore, in this case, \(J\) has \(2N\) negative real eigenvalues, which implies the origin is locally exponentially stable. For \(|\gamma| > 2|\omega|\),  each pair of roots is distinct from each other, so the oscillations are overdamped. For \(|\gamma|<2|\omega|\), if \(\gamma<2ca\), \(J\) has the same \(N-1\) pairs of complex conjugate eigenvalues with negative real parts , \{$\lambda_{1+}$,$\lambda_{1-}$\}, and a pair of complex conjugate eigenvalues, \{$\lambda_{2+}$,$\lambda_{2-}$\}, which cross the imaginary axis due to a variation of \(\gamma\).
	
	Note that
	\begin{align*}
		\left.\text{Re}(\lambda_{2\pm})(\gamma)\right|_{\gamma=0}=&0 \\
		\left.\frac{d\text{Re}(\lambda_{2\pm})(\gamma)}{d\gamma}\right|_{\gamma=0}=&\frac{1}{2}>0, \\
	\end{align*}
	
	\noindent where \(\text{Re}(\lambda_{2\pm})\) denote the real parts of \(\lambda_{2\pm}\). According to the Hopf bifurcation theorem, a Hopf bifurcation arises from the origin at \(\gamma=0\), and the origin is locally exponentially stable for \(\gamma<0\) and unstable for \(\gamma>0\). For \(\gamma<0\), the oscillations are underdamped since \(J\) has \(2N\) complex conjugate eigenvalues with negative real parts in this case. If we further have \(\alpha \geqslant 0\) and \(\beta \geqslant 0\), then the origin is globally asymptotically stable according to Lemma 4.5, which means that there is no limit cycle around the origin for \(\gamma<0\). Thus, the Hopf bifurcation must be supercritical. 
\end{proof}

By the Hopf bifurcation theorem for multi-dimensional systems, for sufficiently small values of \(\gamma \neq 0\) a continuous family of stable limit cycles \(\{\xi(\gamma)\}\), with the parameter \(\gamma\), bifurcates from the origin into the region \(\gamma>0\) in a \(2N\) dimensional state space when the supercritical Hopf bifurcation occurs in the coupled-oscillator system (\ref{38*}). By an empirical rule of thumb for supercritical Hopf bifurcations, the size of \(\xi(\gamma)\) grows continuously from zero and increases proportionally to \(\sqrt{\gamma}\). For fixed values of \(\alpha\) and \(\beta\) and a sufficiently small variable \(\gamma>0\), the ``projection" of the limit cycle \(\xi(\gamma)\) onto each \(x_iy_i\)-plane is roughly a circle, \(\xi_i(\gamma)\), of radius, \(r_i(\gamma)\propto\sqrt{\gamma}\), around the origin, whose shape is distorted as \(\gamma\) gets larger. It is evident that if a network that starts near its stable limit cycle \(\xi(\gamma)\) is completely synchronized, then all projections \(\xi_i(\gamma)\) must coincide and have the same radius \(r(\gamma)\). For a sufficiently small value of \(\gamma\), the origin is, in fact, the only equilibrium point of the system (\ref{38*}) with \(\alpha \geqslant 0\), \(\beta \geqslant 0\), and \(ca>0\), and the bifurcated limit cycle \(\xi(\gamma)\) is also unique, which makes \(\xi(\gamma)\) a globally stable limit cycle so that every trajectory of (\ref{38*}) that does not start at the origin will spiral into \(\xi(\gamma)\) as time approaches infinity in this case. 

When a globally stable limit cycle \(\xi(\gamma)\) that evolves around the origin exists in a coupled-oscillator system (\ref{38*}), it is plausible to conjecture that if the projected trajectories \((x_j(t),y_j(t))\) of (\ref{38*}) start in a disk of radius \(R\) that contains \(\xi_j(\gamma)\) for all \(j\) from 1 to \(N\), then we have

\begin{equation} \label{45*}
	\sum_{j=1}^{N} x_{j}^2(t)+y_{j}^2(t) \leqslant NR^2.
\end{equation}

\noindent Then we can see from the proof of Lemma \ref{lem4.3} that at any point, \((x,y)\), on the boundary of the disk \({\mathfrak{R}_{RvdP}}_{\omega^2=1}\) of radius \(r\) we have

\begin{equation*}
	\begin{split}
		\nabla H \cdot \langle \dot{x}(t), \dot{y}(t)\rangle=&2x \dot{x}(t)+2y \dot{y}(t) \\
		=&-2\left(\alpha x^{2}+\beta y^2-\gamma\right)y^2 \\
		&+2\frac{ca}{N}(\sum_{j=1}^{N} (xx_{j}(t)+yy_{j}(t))-N(x^2+y^2)) \\
		\leqslant &2\gamma(x^2+y^2)+2\frac{ca}{N}(\sum_{j=1}^{N} (xx_{j}(t)+yy_{j}(t))-N(x^2+y^2)) \\
		= &2\frac{ca}{N}(\sum_{j=1}^{N} rr_{j}(t)\cos \theta(t)-N(1-\frac{\gamma}{ca})r^2) \\
		\leqslant &2\frac{ca}{N}(\sum_{j=1}^{N} rr_{j}(t)-N(1-\frac{\gamma}{ca})r^2).
	\end{split}
\end{equation*}

\noindent Suppose that \({x^2_j}_0+{y^2_j}_0 \leqslant (1-\frac{\gamma}{ca})^2r^2\)  and that \(r_j(\gamma)<(1-\frac{\gamma}{ca})r\) for all \(j\). By the Cauchy–Schwarz inequality and Eq.~(\ref{45*}), we have

\begin{equation*}
	\frac{(\sum_{j=1}^{N} r_{j}(t))^2}{N} \leqslant \sum_{j=1}^{N} r_{j}^2(t) \leqslant N(1-\frac{\gamma}{ca})^2r^2.
\end{equation*}

\noindent Hence,

\begin{equation*}
	\sum_{j=1}^{N} rr_{j}(t) \leqslant N(1-\frac{\gamma}{ca})r^2.
\end{equation*}

\noindent Therefore, \(\nabla H \cdot \langle \dot{x}(t), \dot{y}(t)\rangle \leqslant 0\) on the boundary of \({\mathfrak{R}_{RvdP}}_{\omega^2=1}\). It follows that the disk, \({\mathfrak{R}_{RvdP}}_{\omega^2=1}\), having radius \(r>\max \{r_j(\gamma)\}(1-\frac{\gamma}{ca})^{-1}\) is a ``vague" trapping region of the virtual system (\ref{39*}) with \(\alpha \geqslant 0\), \(\beta \geqslant 0\), \(ca > 0\), \(\omega^2=1\), and with a sufficiently small value of \(\gamma\) satisfying \(0<\gamma < 2\min\{ca,1\}\) in the sence that every trajectory that starts within a smaller disk of radius \((1-\frac{\gamma}{ca})r\) which is included in \({\mathfrak{R}_{RvdP}}_{\omega^2=1}\) will remain in  \({\mathfrak{R}_{RvdP}}_{\omega^2=1}\) as the system (\ref{39*}) evolves, and that the corresponding network (\ref{38*}) is completely synchronized in the cases given in Lemma \ref{lem4.4} if \({x^2_i}_0+{y^2_i}_0 \leqslant (1-\frac{\gamma}{ca})^2r^2\) for all \(i\) from 1 to \(N\) and \(\sum_{i=1}^{N} {x^2_i}_0+{y^2_i}_0 \neq 0\).

We do not prove the inequality (\ref{45*}). However, this is not a crucial problem if we do not limit ourselves to exponential synchronization because any trajectory of (\ref{38*}) that does not start from the origin would eventually get infinitely close to \(\xi(\gamma)\) on which (\ref{45*}) holds for \(R \geqslant \max \{r_j(\gamma)\}\). 	

\begin{example}
	Consider a network of six all-to-all coupled identical RvdP oscillators with an identical full-state linear coupling described by Eq.~(\ref{38*}) with \(\alpha = 1\), \(\beta = 1\), \(\gamma  = 0.6\), \(\omega = 1\), \(a=1\), \(c=1\), and \(N  = 6\). Assuming \(\max \{r_j(\gamma)\}\sim \sqrt{\gamma}\), one may choose the ``vague" trapping region, \({\mathfrak{R}_{RvdP}}_{\omega^2=1}\), to be a disk of radius \(r=2>\sqrt{\gamma}(1-\frac{\gamma}{ca})^{-1} \approx 1.94\). The starting points of this system are
	
	\begin{gather*}
		({x_1}_0,{y_1}_0)=(0.6,0.5), \; ({x_2}_0,{y_2}_0)=(0.3,-0.2), \\
		({x_3}_0,{y_3}_0)=(-0.7,-0.3), \; ({x_4}_0,{y_4}_0)=(0.5,-0.5), \\
		({x_5}_0,{y_5}_0)=(0.1,0.4), \; ({x_6}_0,{y_6}_0)=(0,0).
	\end{gather*}
	
	\noindent We have \({x^2_i}_0+{y^2_i}_0 \leqslant (1-\frac{\gamma}{ca})^2r^2=0.64\) for all \(i\) from 1 to 6, and it is obvious that \(\sum_{i=1}^{6} {x^2_i}_0+{y^2_i}_0 \neq 0\).
	
	By Lemma \ref{lem4.4} (case 4 of the lemma),  \({\mathfrak{R}_{RvdP}}_{\omega^2=1}\) is included in \({\mathfrak{C}_{RvdP}}_{\omega^2=1}\) since
	
	\begin{gather*}
		4=r^2 \geqslant \max\{\frac{ca \alpha-3ca \beta}{\alpha^2}, \frac{3ca \beta-ca \alpha}{\alpha^2}\}=2, \\
		4-2\sqrt{3.4}=\frac{ca \alpha+3ca \beta-2\sqrt{3c^2a^2 \alpha \beta+\alpha^2 ca (ca-\gamma)}}{\alpha^2}< r^2 \\
		< \frac{ca \alpha+3ca \beta+2\sqrt{3c^2a^2 \alpha \beta+\alpha^2 ca (ca-\gamma)}}{\alpha^2}=4+2\sqrt{3.4}, \\
		3ca\alpha\beta+\alpha^2 (ca-\gamma) =3.4 >0.
	\end{gather*}
	
	\noindent Then, by the above analysis, we can predict that this network is completely synchronized, which can be seen in Fig. \ref{fig4.3}.
	
	\begin{figure}[htbp]
		\centering
		\begin{subfigure}{.25\textwidth}
			\centering
			\includegraphics[width=.8\linewidth]{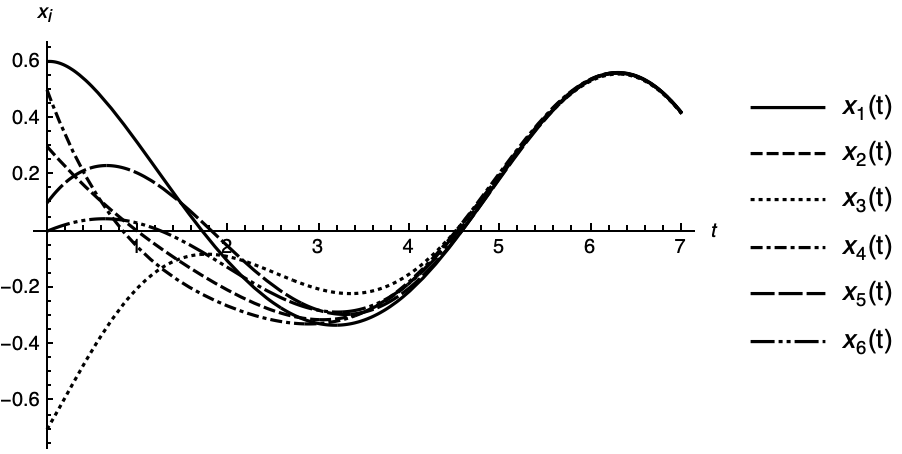}
			\caption{The time series of \(x_i\) for \(i\) from 1 to 6}
		\end{subfigure}%
		\begin{subfigure}{.25\textwidth}
			\centering
			\includegraphics[width=.8\linewidth]{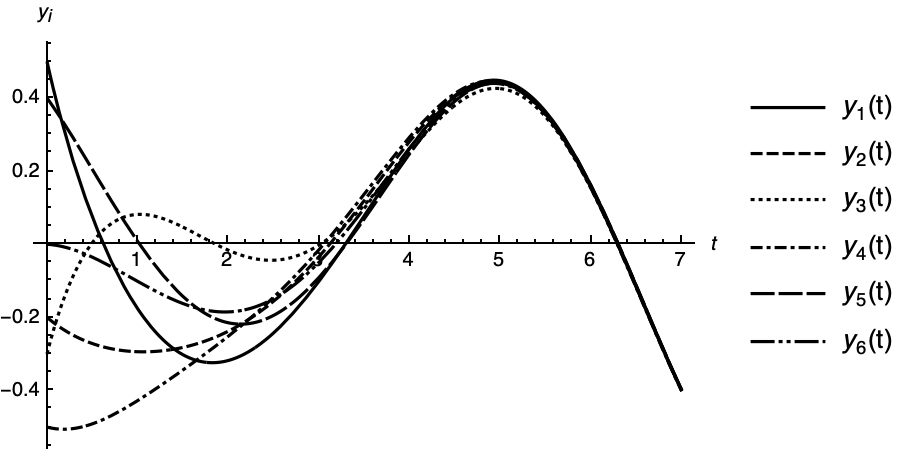}
			\caption{The time series of \(y_i\) for \(i\) from 1 to 6}
		\end{subfigure}
		\begin{subfigure}{.25\textwidth}
			\centering
			\includegraphics[width=.8\linewidth]{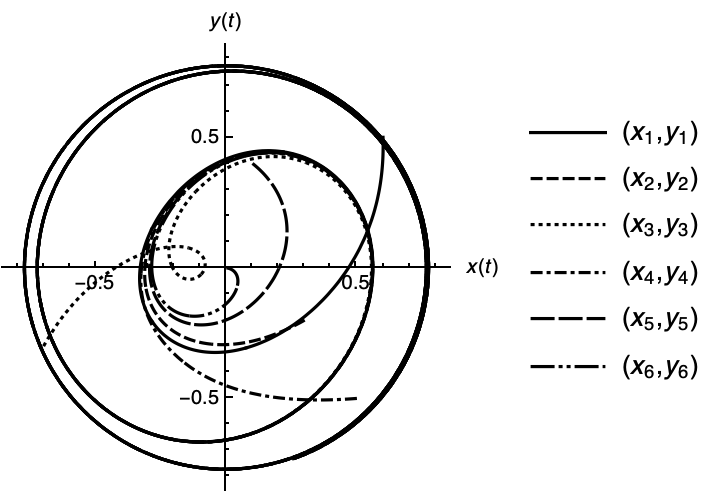}
			\caption{Phase portraits}
		\end{subfigure}
		\caption{The time series and phase portraits of six self-sustained all-to-all coupled identical RvdP oscillators with an identical full-state linear coupling}
		\label{fig4.3}
	\end{figure}	
	
\end{example}

Determining the direction of the Hopf bifurcation in the system (\ref{38*}) at \(\gamma=0\) is nontrivial when \(\alpha\) or \(\beta\) is negative. Moreover, the large number of oscillators in the network makes it even more difficult. For simplicity, we only consider the case of two coupled oscillators, i.e., \(N=2\).

For a system of the form (1) with the bifurcation parameter \(\nu\), let \(J_{\pmb{f}}(\nu)\) be the Jacobian at the equilibrium point \(\pmb{x}_{*}(\nu)\), i.e., \(J_{\pmb{f}}(\nu)=\frac{\partial \pmb{f}(\pmb{x};\nu)}{\partial {\pmb{x}}}|_{\pmb{x}=\pmb{x}_{*}}\). Suppose \(J_{\pmb{f}}(\nu)\) has \(m\) pairs of complex conjugate eigenvalues, \(\{\lambda_{i+}, \lambda_{i-}\}\), where \(\lambda_{i+}\) have positive imaginary parts and \(\lambda_{i-}\) have negative ones, and \(n-2m\) real eigenvalues, \({\lambda_{r}}_j\). Let \(\text{Re}(\lambda_{1+}) \geqslant \text{Re}(\lambda_{2+}) \geqslant \cdots \geqslant \text{Re}(\lambda_{m+})\) and let \({\lambda_r}_1 \geqslant {\lambda_r}_2 \geqslant \cdots \geqslant {\lambda_r}_{n-2m}\). The system can be rewritten as \(\dot{\pmb{x}}=\pmb{f}(\pmb{x};\nu)=J_{\pmb{f}}(\nu)\pmb{x}+\pmb{h}(\pmb{x})\), where \(\pmb{h}(\pmb{x})\) is the nonlinear term.

The Hopf bifurcation theorem gives that the Hopf bifurcation occurs at \(\nu_{c}\) where \(\nu_{c}\) satiesfies \(\text{Re}(\lambda_{1+}(\nu_{c}))=0\), \(\frac{d\text{Re}(\lambda_{1+}(\nu_c))}{d\nu} \neq 0\), \(\text{Im}(\lambda_{1+}(\nu_{c})) \neq 0\), \(\text{Re}(\lambda_{i+}(\nu_{c}))<0\) for \(i\) from 2 to \(m\), and \({\lambda_r}_j(\nu_c)<0\) for \(j\) from 1 to \(n-2m\). Form a matrix \(P\equiv(\text{Re}(\pmb{v}_{1+}) \; -\text{Im}(\pmb{v}_{1+}) \; \cdots \; \text{Re}(\pmb{v}_{m+}) \; -\text{Im}(\pmb{v}_{m+}) \; {\pmb{v}_{r}}_1 \; \cdots \; {\pmb{v}_r}_{n-2m})\), where \(\pmb{v}_{i+}\) and \({\pmb{v}_{r}}_{j}\) are (generalized) eigenvectors of \(J_{\pmb{f}}(\nu_c)\) corresponding to \(\lambda_{i+}(\nu_{c})\) and \({\lambda_r}_{j}(\nu_{c})\) respectively. We have the following lemma\cite{hassard}.

\begin{lemma}\label{lem4.7}
	\(P\) is invertable and \(P^{-1}J_{\pmb{f}}(\nu_{c})P\) is in the Jordan canonical form, i.e., 
	\begin{equation} \label{44*}
		P^{-1}J_{\pmb{f}}(\nu_{c})P=\left(
		\begin{array}{r@{}c|c@{}l}
			&    \begin{smallmatrix}
				0 & -\text{Im}(\lambda_1(\nu_{c})) \\
				\text{Im}(\lambda_1(\nu_{c})) & 0\rule[-1ex]{0pt}{2ex}
			\end{smallmatrix} & \mbox{\large0} \\\hline
			&    \mbox{\large0} &  \mbox{D}
		\end{array} 
		\right),
	\end{equation}
	
	\noindent where \(D\) is a block diagonal matrix of original dimension \(n-2\), whose diagonals are Jordan blocks \(D_i\) that are either of the form
	
	\begin{equation*}
		\left(\begin{array}{cccc}
			B_i & I_{2 \times 2} & & \\
			& B_i & \ddots &  \\
			& & \ddots & I_{2 \times 2} \\
			& & & B_i \\
		\end{array}\right)
	\end{equation*}
	
	\noindent for \(i\) from 2 to \(m\) with
	
	\begin{equation*}
		B_i=\left(\begin{array}{cc}
			\text{Re}(\lambda_{i+}(\nu_{c}))&-\text{Im}(\lambda_{i+}(\nu_{c})) \\
			\text{Im}(\lambda_{i+}(\nu_{c}))  & \text{Re}(\lambda_{i+}(\nu_{c}))  \\
		\end{array}\right)
	\end{equation*}
	
	\noindent for complex eigenvalues \(\lambda_{i+}(\nu_{c})\), or of the form
	
	\begin{equation*}
		\left(\begin{array}{cccc}
			{\lambda_r}_j(\nu_{c}) & 1 & & \\
			& {\lambda_r}_j(\nu_{c})  & \ddots &  \\
			& & \ddots & 1 \\
			& & & {\lambda_r}_j(\nu_{c})  \\
		\end{array}\right),
	\end{equation*}
	
	\noindent for real eigenvalue \({\lambda_r}_j(\nu_{c})\) with \(j\) from 1 to \(n-2m\).
\end{lemma}

Perforem coordinate transformations from \(x\) to \(y\) frames by \(\pmb{x}=P\pmb{y}+\pmb{x}_{*}(\nu_c)\). Then the system at \(\nu=\nu_c\) is converted into

\begin{equation*}
	\begin{split}
		\dot{\pmb{y}}=P^{-1}\dot{\pmb{x}}=&P^{-1}\pmb{f}(P\pmb{y}+\pmb{x}_{*}(\nu_c)) \\
		=&P^{-1}J_{\pmb{f}}(\nu_c)P\pmb{y}+P^{-1}J_{\pmb{f}}(\nu_c)\pmb{x}_{*}(\nu_c)+\pmb{h}(P\pmb{y}+\pmb{x}_{*}(\nu_c)) \\
		\equiv &\pmb{F}(\pmb{y}),
	\end{split}
\end{equation*}

\noindent where \(\pmb{F}=\langle F_{1}, \cdots, F_{n}\rangle:\mathbb{R}^n\to \mathbb{R}^n\) is a vector-valued function. Then \(\frac{\partial \pmb{F}(\pmb{y})}{\partial {\pmb{y}}}|_{\pmb{y}=\pmb{0}}=P^{-1}J_{\pmb{f}}(\nu_c)P\). Thus, by Lemma \ref{lem4.7}, the Jacobian for \(\pmb{y}\)-system at \(\pmb{y}=\pmb{0}\) is in the Jordan canonical form. We calculate the following quantities at \(\pmb{y}=\langle y_{1}, \cdots, y_{n}\rangle=\pmb{0}\).

\begin{widetext}
	\begin{equation*}
		\begin{split}
			g_{11}=&\frac{1}{4}[\frac{\partial^2 F_1}{\partial {y_1}^2}+\frac{\partial^2 F_1}{\partial {y_2}^2}+i(\frac{\partial^2 F_2}{\partial {y_1}^2}+\frac{\partial^2 F_2}{\partial {y_2}^2})], \\
			g_{20}=&\frac{1}{4}[\frac{\partial^2 F_1}{\partial {y_1}^2}-\frac{\partial^2 F_1}{\partial {y_2}^2}+2\frac{\partial^2 F_2}{\partial y_1 \partial y_2}+i(\frac{\partial^2 F_2}{\partial {y_1}^2}-\frac{\partial^2 F_2}{\partial {y_2}^2}-2\frac{\partial^2 F_1}{\partial y_1 \partial y_2})], \\
			g_{02}=&\frac{1}{4}[\frac{\partial^2 F_1}{\partial {y_1}^2}-\frac{\partial^2 F_1}{\partial {y_2}^2}-2\frac{\partial^2 F_2}{\partial y_1 \partial y_2}+i(\frac{\partial^2 F_2}{\partial {y_1}^2}-\frac{\partial^2 F_2}{\partial {y_2}^2}+2\frac{\partial^2 F_1}{\partial y_1 \partial y_2})], \\
			G_{21}=&\frac{1}{8}[\frac{\partial^3 F_1}{\partial {y_1}^3}+\frac{\partial^3 F_1}{\partial {y_1} \partial {y_2}^2}+\frac{\partial^3 F_2}{\partial {y_1}^2 \partial {y_2}}+\frac{\partial^3 F_2}{\partial {y_2}^3}+i(\frac{\partial^3 F_2}{\partial {y_1}^3}+\frac{\partial^3 F_2}{\partial {y_1} \partial {y_2}^2}-\frac{\partial^3 F_1}{\partial {y_1}^2 \partial {y_2}}-\frac{\partial^3 F_1}{\partial {y_2}^3})], \\
			g_{21}=&G_{21}+\sum_{i=1}^{n-2} (2G^{i}_{110}w^{i}_{11}+G^{i}_{101}w^{i}_{20}),
		\end{split}
	\end{equation*}
\end{widetext}
\noindent where
\begin{equation*}
	\begin{split}
		G^{j-2}_{110}=&\frac{1}{2}[\frac{\partial^2 F_1}{\partial y_1 \partial y_j}+\frac{\partial^2 F_2}{\partial y_2 \partial y_j}+i(\frac{\partial^2 F_2}{\partial y_1 \partial y_j}-\frac{\partial^2 F_1}{\partial y_2 \partial y_j})], \\
		G^{j-2}_{101}=&\frac{1}{2}[\frac{\partial^2 F_1}{\partial y_1 \partial y_j}-\frac{\partial^2 F_2}{\partial y_2 \partial y_j}+i(\frac{\partial^2 F_2}{\partial y_1 \partial y_j}+\frac{\partial^2 F_1}{\partial y_2 \partial y_j})]
	\end{split}
\end{equation*}
\noindent for \(j=3,\cdots, n\), and where \(\pmb{w}_{11}=\langle w^{1}_{11}, \cdots, w^{n-2}_{11}\rangle\) and \(\pmb{w}_{20}=\langle w^{1}_{20}, \cdots, w^{n-2}_{20}\rangle\) are the solutions of 

\begin{equation*}
	\begin{split}
		&D\pmb{w}_{11}=-\pmb{h}_{11}, \\
		(&D-2\text{Im}(\lambda_1(\nu_{c})) i I_{n-2})\pmb{w}_{20}=-\pmb{h}_{20},
	\end{split}
\end{equation*}

\noindent respectively, with \(\pmb{h}_{11}=\langle h^{1}_{11}, \cdots, h^{n-2}_{11}\rangle\) and \(\pmb{h}_{20}=\langle h^{1}_{20}, \cdots, h^{n-2}_{20}\rangle\) whose elements are defined by

\begin{equation*}
	\begin{split}
		h^{j-2}_{11}=&\frac{1}{4}(\frac{\partial^2 F_j}{\partial {y_1}^2}+\frac{\partial^2 F_j}{\partial {y_2}^2}), \\
		h^{j-2}_{20}=&\frac{1}{4}(\frac{\partial^2 F_j}{\partial {y_1}^2}-\frac{\partial^2 F_j}{\partial {y_2}^2}-2i\frac{\partial^2 F_j}{\partial y_1 \partial y_2})
	\end{split}
\end{equation*}

\noindent for \(j=3,\cdots, n\), and with matrix \(D\) defined in Eq.~(\ref{44*}). We then calculate

\begin{equation*}
	\beta_2=\text{Re}[\frac{i}{\text{Im}(\lambda_1(\nu_{c}))}(g_{20}g_{11}-2|g_{11}|^2-\frac{|g_{02}|^2}{3})+g_{21}].
\end{equation*}

The limit cycle is stable (unstable) if \(\beta_2<0\) (\(\beta_2>0\)) on the side of \(\nu=\nu_{c}\) where the limit cycle appears, which indicates the Hoft bifurcation is supercritical (subcritical)\cite{hassard}. 

For system (\(\ref{38*}\)) with \(N=2\), we have

\begin{gather*}
	\pmb{f}(\pmb{x};\gamma)=\left(\begin{array}{c}
		x_3+\frac{ca(x_2-x_1)}{2} \\
		x_4+\frac{ca(x_1-x_2)}{2} \\
		-(\alpha {x_1}^2+\beta {x_3}^2-\gamma)x_3-\omega^2 x_1+\frac{ca(x_4-x_3)}{2} \\
		-(\alpha {x_2}^2+\beta {x_4}^2-\gamma)x_4-\omega^2 x_2+\frac{ca(x_3-x_4)}{2} \\
	\end{array}\right), \\
	P=\left(\begin{array}{cccc}
		0 & {|\omega|}^{-1} & 0 & -{|\omega|}^{-1} \\
		0 & {|\omega|}^{-1}& 0 & {|\omega|}^{-1} \\
		1 & 0 & -1 & 0 \\
		1 & 0 & 1 & 0 \\
	\end{array}\right), \\
	g_{11}=g_{20}=g_{02}=0, \\
	G_{21}=-\frac{\alpha+3\beta \omega^2}{4\omega^2}, \\
	\pmb{h}_{11}=\pmb{h}_{20}=\pmb{0} , \\
	\pmb{w}_{11}=\pmb{w}_{20}=\pmb{0}, \\
	g_{21}=G_{21}, 
\end{gather*}

\noindent and
\begin{equation*}
	\beta_2=-\frac{\alpha+3\beta \omega^2}{4\omega^2}.
\end{equation*}

If \(\omega \neq 0\) and \(\alpha+3\beta \omega^2>0\) (\(<0\)), then \(\beta_2<0\) (\(>0\)), and thus a supercritical (subcritical) Hopf bifurcation occurs. For systems with more oscillators, we can analyze the direction of the Hopf bifurcation in the same way, which will not be discussed in this paper. Combined with Lemma \ref{lem4.6}, we conclude

\begin{theorem}\label{thm4.5}
	For a two-oscillator system described by Eq.~(\ref{38*}) with \(ca>0\) and \(N=2\), when \(|\gamma| < 2\omega\) and \(\gamma<2ca\), the origin is locally exponentially stable for \(\gamma<0\) and unstable for \(\gamma>0\), and a continuous one-parameter family of stable (unstable) limit cycles bifurcate from the origin into the region \(\gamma>0\) (\(<0\)) if \(\alpha+3\beta \omega^2>0\) (\(<0\)). 
\end{theorem}

In the case of a supercritical Hoft bifurcation, the conclusion about the conditions for synchronization of the network (\ref{38*}) with \(\alpha,\beta \geqslant 0\) also applies here for the case \(\alpha+3\beta \omega^2>0\) with a negative value of \(\alpha\) or \(\beta\). However, there are two issues to be noted, one is that the shape of the limit cycle becomes distorted more rapidly from a circle with increasing \(\gamma\) compared to the case for \(\alpha,\beta \geqslant 0\), and the other is that the stability of the limit cycle becomes local since the origin may not be a unique equilibrium point even for a minimal value of \(\gamma\). Therefore, in the case where either \(\alpha<0\) or \(\beta<0\), the radius of a vague trapping region \({\mathfrak{R}_{RvdP}}_{\omega^2=1}\) of the virtual system (\ref{39*}) can be taken to be only slightly greater than \(\max\{r_j(\gamma)\}(1-\frac{\gamma}{ca})^{-1}\).

\appendix

\section{\label{appA}Proofs of Propositions \ref{prop3.1} and \ref{prop3.2}}

\subsection{\label{appA:subsec1}Proof of Proposition \ref{prop3.1}}

\begin{proof}
	Notice that the virtual system of (\ref{5*}) can be constructed as follows
	\begin{equation} \label{5}
		\dot{\pmb{y}}=\pmb{f}\left(\pmb{y}\right)+\sum_{j=1}^{N} A_j\pmb{x}_j-(\sum_{j=1}^{N} A_j)\pmb{y}.
	\end{equation}
	We have
	\begin{equation*}
		\frac{\partial \pmb{\Phi}}{\partial \pmb{y}}=J_{\pmb{f}}-\sum_{j=1}^{N} A_j.
	\end{equation*}
	By Definition \ref{def2.1}, the contractoin region of (\ref{5}) is the region in which \((\frac{\partial \pmb{\Phi}}{\partial \pmb{y}})_{s}={J_{\pmb{f}}}_s-\sum_{j=1}^{N} {A_j}_s\prec 0\). Therefore, by Theorem \ref{thm3.1}, if \(\pmb{x}_{1}(t)\) and \(\pmb{x}_{2}(t)\) remain within this region, (\ref{5*}) is completely synchronized.
\end{proof}

\subsection{\label{appA:subsec2}Proof of Proposition \ref{prop3.2}}

\begin{proof}
	Notice that the network (\ref{8}) does not have a single virtual system, but it is easy to see that the virtual system of (\ref{8b}) is
	\begin{equation} \label{9}
		\dot{\pmb{y}}=\pmb{f}\left(\pmb{y}\right)+A_1(\pmb{x}_1-\pmb{y})
	\end{equation}
	Since \(\pmb{x}_{i}(t)\) remain in the contraction region of (\ref{9}) where \({J_{\pmb{f}}}_s-{A_{1}}_s \prec 0\) at all times for all \(i\) from 2 to \(N\), by Theorem \ref{thm3.1}, all the oscillators except the first one converge toward the same synchronous state exponentially. Hence Eq.~(\ref{8a}) will reduce to 
	\begin{equation*}
		\dot{\pmb{x}}_1=\pmb{f}\left(\pmb{x}_1\right)+\sum_{j=2}^{N} A_j(\pmb{x}_2-\pmb{x}_1),
	\end{equation*}
	\noindent which describes the state of the first oscillator in the network. Then, the virtual system of the subnetwork consisting of only the first two oscillators is
	\begin{equation} \label{10}
		\dot{\pmb{y}}=\pmb{f}\left(\pmb{y}\right)+A_1\pmb{x}_1+(\sum_{j=2}^{N} A_j)\pmb{x}_2-(\sum_{j=1}^{N} A_j)\pmb{y}.
	\end{equation}
	Therefore, the first two oscillators are completely synchronized as \(\pmb{x}_1(t)\) and \(\pmb{x}_2(t)\) are always inside the contraction region of (\ref{10}) where \({J_{f}}_s-\sum_{j=1}^{N} {A_{j}}_s \prec 0\), and thus the whole newtork (\ref{8}) is completely synchronized by transitivity. 
\end{proof}

\section{\label{appB}Proofs of Lemmas \ref{lem4.1}-\ref{lem4.5}}

\subsection{\label{appB:subsec1}Proof of Lemma \ref{lem4.1}}

\begin{proof}
	Since the coupling is all-to-all and symmetric, Proposition \ref{prop3.1} implies that the network is completely synchronized if \((x_{i}(t)\text{  }y_{i}(t))^T\) are always in  the region where \({J_{\pmb{f}}}_s-\sum_{j=1}^{N} {A_{j}}_s \prec 0\) for all \(i\). Using Eq.~(\ref{37*}) and Schur complement lemma, \({J_{\pmb{f}}}_s-\sum_{j=1}^{N} {A_{j}}_s \prec 0\) iff \(-{J_{\pmb{f}}}_s+\sum_{j=1}^{N} {A_{j}}_s \succ 0\) iff \(\frac{c}{N}\sum_{j=1}^{N} a_{j}>0\) and \((\alpha x^2+3\beta y^2-\gamma)+\frac{c}{N}\sum_{j=1}^{N} b_{j}-\frac{(1-2 \alpha  x y-\omega^2)^2}{4}(\frac{c}{N}\sum_{j=1}^{N} a_{j})^{-1}>0\).
\end{proof}

\subsection{\label{appB:subsec2}Proof of Lemma \ref{lem4.2}}

\begin{proof}
	Note that the origin is an equilibrium point. Since \(\omega \neq 0\), we can construct a continuously
	differentiable function \(V(x_{i},y_{i})=\sum_{i=1}^{N} (\frac{{x_{i}}^2}{2}+\frac{{y_{i}}^2}{2\omega^2})\). This is a  Lyapunov function for the system (\ref{38*}). To see this, we take the derivative of \(V\) with respect to time along an arbitrary trajectory of Eq.~(\ref{38*}).
	
	\begin{equation*}
		\begin{split}
			\dot{V}(x_{i},y_{i})=&\sum_{i=1}^{N} (x_{i}\dot{x}_{i}+\frac{y_{i}\dot{y}_{i}}{\omega^2}) \\
			=&\sum_{i=1}^{N} (x_{i}y_{i}-\frac{ca}{N}x_{i}\sum_{j=1}^{N} (x_{i}-x_{j}) \\
			&-\frac{y_{i}}{\omega^2}\left(\alpha {x_{i}}^{2}+\beta {y_{i}}^2-\gamma\right) y_{i}-\frac{y_{i}}{\omega^2}\omega^{2}x_{i}-\frac{ca}{N\omega^2}y_{i}\sum_{j=1}^{N} (y_{i}-y_{j})) \\
			\leqslant &\frac{\gamma}{\omega^2}\sum_{i=1}^{N} {y_{i}}^2-\frac{ca}{N}\sum_{i,j=1}^{N} x_{i}(x_{i}-x_{j})-\frac{ca}{N\omega^2}\sum_{i,j=1}^{N} y_{i}(y_{i}-y_{j}) \\
			=&\frac{\gamma}{\omega^2}\sum_{i=1}^{N} {y_{i}}^2-\frac{ca}{2N}\sum_{i,j=1}^{N} ({x_{i}}^2-2x_{i}x_{j}+{x_{i}}^2) \\
			&-\frac{ca}{2N\omega^2}\sum_{i,j=1}^{N} ({y_{i}}^2-2y_{i}y_{j}+{y_{i}}^2).
		\end{split}
	\end{equation*}
	
	Notice that \(\sum_{i,j=1}^{N} {x_{i}}^2=\sum_{j,i=1}^{N} {x_{j}}^2\) and that \(\sum_{i,j=1}^{N} {y_{i}}^2=\sum_{j,i=1}^{N} {y_{j}}^2\). We therefore have
	\begin{equation*}
		\begin{split}
			\dot{V}(x_{i},y_{i})\leqslant&\frac{\gamma}{\omega^2}\sum_{i=1}^{N} {y_{i}}^2-\frac{ca}{2N}\sum_{i,j=1}^{N} ({x_{i}}^2-2x_{i}x_{j}+{x_{j}}^2) \\
			&-\frac{ca}{2N\omega^2}\sum_{i,j=1}^{N} ({y_{i}}^2-2y_{i}y_{j}+{y_{j}}^2) \\
			\leqslant&\frac{\gamma}{\omega^2}\sum_{i=1}^{N} {y_{i}}^2-\frac{ca}{2N}\sum_{i,j=1}^{N} (x_{i}-x_{j})^2-\frac{ca}{2N\omega^2}\sum_{i,j=1}^{N} (y_{i}-y_{j})^2 \\
			\leqslant&\frac{\gamma}{\omega^2}\sum_{i=1}^{N} {y_{i}}^2\leqslant 0.
		\end{split}
	\end{equation*}
	
	\noindent Since \(\gamma<0\), the equality holds if and only if \(\pmb{x}(t)=(x_{i}(t)\text{ }y_{i}(t))^T=\pmb{0}\) for \(t\geqslant 0\). Note that
	
	\begin{equation*}
		\begin{split}
			V(\pmb{0})=&0 \\
			V(\pmb{x})>&0\text{, }\forall \pmb{x}\neq \pmb{0} \\
			V(\pmb{x})\rightarrow&\infty\text{, as }\|x\|\rightarrow \infty
		\end{split}
	\end{equation*}
	\noindent Thus, by LaSalle's invariance principle, the origin is globally asymptotically stable.
\end{proof}

\subsection{\label{appB:subsec3}Proof of Lemma \ref{lem4.3}}

\begin{proof}
	Recall that the gradient of \(H\) at a point on a level curve of \(H\) is perpendicular to the level curve and points toward the fastest increase. To verify that the vector field of (\ref{39*}) points inward on the boundary of the disk \({\mathfrak{R}_{RvdP}}_{\omega^2=1}\) at all times, it suffices to show that \(\nabla H \cdot \langle \dot{x}, \dot{y}\rangle = 2x \dot{x}+2y \dot{x} \leqslant 0\) on the circle \(x^2+y^2 = r^2\) at all times. At any point \((x,y)\) on \(x^2+y^2 = r^2\), we have
	\begin{equation*}
		\begin{split}
			&\nabla H \cdot \langle \dot{x}(t), \dot{y}(t)\rangle=2x \dot{x}(t)+2y \dot{y}(t) \\
			&=2xy+2\frac{ca}{N}(\sum_{j=1}^{N} xx_{j}(t)-Nx^2) \\
			&-2\left(\alpha x^{2}+\beta y^2-\gamma\right)y^2-2xy+2\frac{ca}{N}(\sum_{j=1}^{N} yy_{j}(t)-Ny^2) \\
			&\leqslant 2\frac{ca}{N}(\sum_{j=1}^{N} (xx_{j}(t)+yy_{j}(t))-N(x^2+y^2)) \\
			&=2\frac{ca}{N}(\sum_{j=1}^{N} rr_{j}(t)\cos \theta_i(t)-Nr^2) \\
			&\leqslant 2\frac{ca}{N}(\sum_{j=1}^{N} rr_{j}(t)-Nr^2),
		\end{split}
	\end{equation*}
	
	\noindent where \(r_{j}(t)\) represent the distances between the origin and the point \((x_j(t),y_j(t))\) at time \(t\) in the phase plane of (\ref{39*}) and \(\theta_i(t)\) represent the angles between the position vectors \(\langle x, y\rangle\) and \(\langle x_j(t), y_j(t)\rangle\). By Lemma \ref{lem4.2}, the origin is a globally asymptotically stable point of the original network, which means
	
	\begin{equation*}
		\sum_{j=1}^{N} r_{j}^2(t) \leqslant \sum_{j=1}^{N} r_{j}^2(0),
	\end{equation*}
	
	\noindent for all \(t>0\). Suppose that the starting points  \(({x_j}_0,{y_j}_0)\) of all oscillators lie in \({\mathfrak{R}_{RvdP}}_{\omega^2=1}\), or, in other words, \(r_j(0) \leqslant r\) for all \(j\) from 1 to \(N\). Then
	
	\begin{equation*}
		\sum_{j=1}^{N} r_{j}^2(t) \leqslant \sum_{j=1}^{N} {x^2_j}_0+{y^2_j}_0 \leqslant Nr^2.
	\end{equation*}
	
	\noindent Using the Cauchy–Schwarz inequality, we have
	
	\begin{equation*}
		\frac{(\sum_{j=1}^{N} r_{j}(t))^2}{N} \leqslant \sum_{j=1}^{N} r_{j}^2(t) \leqslant Nr^2.
	\end{equation*}
	
	\noindent It follows that
	
	\begin{equation*}
		\sum_{j=1}^{N} r_{j}(t) \leqslant Nr,
	\end{equation*}
	
	\noindent or,
	
	\begin{equation*}
		\sum_{j=1}^{N} rr_{j}(t) \leqslant Nr^2,
	\end{equation*}
	
	\noindent Therefore,
	
	\begin{equation*}
		\nabla H \cdot \langle \dot{x}(t), \dot{y}(t)\rangle \leqslant 2\frac{ca}{N}(\sum_{j=1}^{N} rr_{j}(t)-Nr^2) \leqslant 0
	\end{equation*}
\end{proof}

\subsection{\label{appB:subsec4}Proof of Lemma \ref{lem4.4}}

\begin{proof}
	
	Since \(ca >0\) and \(ca > \gamma\), \(c^{2}a^{2}-ca\gamma=ca(ca-\gamma) >0\). Take an arbitrary point \((x^*,y^*)\) in \({\mathfrak{R}_{RvdP}}_{\omega^2=1}\). We want to show \(g(x^*,y^*) > 0\) so that \((x^*,y^*) \in {\mathfrak{C}_{RvdP}}_{\omega^2=1}\).
	
	\textit{Case 1}:  \(3 \beta r^2>\gamma-ca\) and \(\alpha = 0\). In this case,
	
	\begin{equation*}
		g(x^*,y^*)=G(y^*)=ca(ca-\gamma)+3ca \beta  {y^*}^2 =ca(ca-\gamma+3\beta  {y^*}^2),
	\end{equation*}
	
	\noindent where \(y^* \in [-r,r]\). For \(\beta \geqslant 0\), \(G(y)\) has a minimum at \(y=0\), and we have \(G(0)=ca(ca-\gamma)>0\). For \(\beta < 0\) and \(y \in [-r,r]\), \(G(y)\) has minima at \(y=\pm r\). Since \(3 \beta r^2>\gamma-ca\), \(G(\pm r)=ca(ca-\gamma+3\beta  r^2)>0\). Consequently, \(G(y)>0\) for all \(y \in [-r,r]\). Therefore, \(g(x^*,y^*)=G(y^*) > 0\).

	\textit{Case 2}: \(r^2 \leqslant \frac{ca\alpha-3ca\beta}{\alpha^2}\) and \(3 \beta r^2>\gamma-ca\) with \(0 \neq \alpha>3\beta\). 
	
	The partial derivatives of \(g(x,y)\) are
	\begin{equation*}
		\begin{split}
			g_{x}(x,y)&=2ca\alpha x-2\alpha^2 y^{2}x \\
			g_{y}(x,y)&=6ca\beta y-2\alpha^2 x^{2}y.
		\end{split}
	\end{equation*}
	
	\noindent Notice that \(g_{x}(x,y)=0\) when \(x=0\) or \(y^2=\frac{ca}{\alpha}\) (if \(\alpha>0\)) and that \(g_{y}(x,y)=0\) when \(y=0\) or \(x^2=\frac{3ca\beta}{\alpha^2}\) (if \(\alpha \neq 0\) and \(\beta \geqslant 0\)). Thus, for \(\beta<0\), \(g(x,y)\) only has one critical point, \((0,0)\), at which \(g(0,0)=ca(ca-\gamma) >0\). For \(\beta \geqslant 0\), \(g(x,y)\) has five critical points, \((0,0)\), \((\pm \sqrt{\frac{3ca\beta}{\alpha^2}}, \pm \sqrt{\frac{ca}{\alpha}})\), if \(\alpha>0\). However, 
	
	\begin{equation*}
		r^2 \leqslant \frac{ca\alpha-3ca\beta}{\alpha^2} \leqslant \frac{ca\alpha+3ca\beta}{\alpha^2},
	\end{equation*}
	
	\noindent which indicates that there is only one critical point, \((0,0)\), in the interior of \({\mathfrak{R}_{RvdP}}_{\omega^2=1}\) at which \(g(0,0)>0\).
	
	On the boundary, \(x^2+y^2=r^2\), of \({\mathfrak{R}_{RvdP}}_{\omega^2=1}\), we have
	
	\begin{gather*}
		g(x,y)=h(x)	\equiv \\
		ca(ca-\gamma)+ca\alpha  x^2+3ca \beta (r^2-x^2)- \alpha^2  x^2 (r^2-x^2) \\
		=\alpha^2  x^4+(ca\alpha-3ca \beta-\alpha^2 r^2)x^2+ca(ca-\gamma)+3ca \beta r^2 > 0
	\end{gather*}
	\noindent and
	\begin{equation*}
		h^\prime(x)=4\alpha^2  x^3+2(ca\alpha-3ca \beta-\alpha^2 r^2)x.
	\end{equation*}
	
	\noindent Since \(r^2 \leqslant \frac{ca \alpha-3ca \beta}{\alpha^2}\), we have \(ca\alpha-3ca \beta-\alpha^2 r^2 \geqslant 0\). So \(h^{\prime}(x)=0\) only at \(x=0\). Furthermore, \(h^{\prime}(x)>0\) when \(x>0\) and \(h^{\prime}(x)<0\) when \(x<0\). It follows that \(h(x)\) achieves a minimum at \(x=0\). 
	
	\begin{equation*}
		h(0)=ca(ca-\gamma)+3 ca \beta r^2=ca(ca-\gamma+3 \beta).
	\end{equation*}
	
	\noindent Since \(3 \beta r^2>\gamma-ca\), \(h(0)>0\). Then \(h(x)>0\) for all \(x\). Thus, \(g(x^*,y^*) > 0\).
	
	\textit{Case 3}: \(r^2 \leqslant \frac{3ca \beta-ca \alpha}{\alpha^2}\) and \(\alpha r^2>\gamma-ca\) with \(0 \neq \alpha<3\beta\). 
	
	For \(\alpha<0\), \(g(x,y)\) has only one critical point \((0,0)\) at which \(g(0,0)=ca(ca-\gamma) >0\). For \(\alpha>0\),
	
	\begin{equation*}
		r^2 \leqslant \frac{3ca \beta-ca \alpha}{\alpha^2} < \frac{ca\alpha+3ca\beta}{\alpha^2}.
	\end{equation*}
	
	\noindent Only one critical point in \({\mathfrak{R}_{RvdP}}_{\omega^2=1}\) is still \((0,0)\), and \(g(0,0)>0\).
	
	On the boundary of \({\mathfrak{R}_{RvdP}}_{\omega^2=1}\), \(g(x,y)=h(x)\) for \(x \in [-r,r]\). Since \(\alpha<3\beta\), \(h^{\prime}(x)=0\) at \(x=0\) and at
	
	\begin{equation*}
		x=x_{\pm}=\pm\frac{\sqrt{-2(ca\alpha-3ca \beta-\alpha^2 r^2)}}{2 \alpha}.
	\end{equation*}
	
	\noindent We have
	
	\begin{equation*}
		h^{\prime \prime}(x)=12\alpha^2  x^2+2(ca\alpha-3ca \beta-\alpha^2 r^2),
	\end{equation*}
	
	\noindent so that
	
	\begin{align*}
		h^{\prime \prime}(0)=&2(ca\alpha-3ca \beta-\alpha^2 r^2) <0, \\
		h^{\prime \prime}(x_{1\pm})=&-4(ca\alpha-3ca \beta-\alpha^2 r^2) >0.
	\end{align*}
	
	\noindent It follows that  \(h(x)\) reaches its global minimum value at \(x=x_{+}\) or at \(x=x_{-}\). But since \(r^2 \leqslant \frac{3ca \beta-ca \alpha}{\alpha^2}\), 
	
	\begin{gather*}
		2 \alpha^2 r^2 \leqslant -2(ca \alpha-3ca \beta) \\
		4 \alpha^2 r^2 \leqslant -2(ca \alpha-3ca \beta-\alpha^2 r^2) \\
		r^2 \leqslant \frac{-2(ca \alpha-3ca \beta-\alpha^2 r^2)}{4\alpha^2}=x^2_{\pm}
	\end{gather*}
	
	\noindent Therefore, for \(x \in [-r,r]\), \(h(x)\) achieves its minimum at \(x=r\) or at \(x=-r\). Since \(\alpha r^2>\gamma-ca\),
	
	\begin{equation*}
		h(\pm r)=ca\alpha r^2+ca(ca-\gamma)=ca(\alpha r^2+ca-\gamma) > 0.
	\end{equation*}
	
	\noindent So \(h(x)>0\) for all \(x \in [-r,r]\). Therefore, \(g(x^*,y^*) > 0\).
	
	\textit{Case 4}: \(r^2 \geqslant \max\{\frac{ca \alpha-3ca \beta}{\alpha^2}, \frac{3ca \beta-ca \alpha}{\alpha^2}\}\) and \(\frac{ca \alpha+3ca \beta-2\sqrt{3c^2a^2 \alpha \beta+\alpha^2 ca (ca-\gamma)}}{\alpha^2}< r^2 < \frac{ca \alpha+3ca \beta+2\sqrt{3c^2a^2 \alpha \beta+\alpha^2 ca (ca-\gamma)}}{\alpha^2}\) with \(3ca\alpha\beta+\alpha^2 (ca-\gamma) > 0\).
	
	Notice from \(3ca\alpha\beta+\alpha^2 (ca-\gamma) > 0\) that \(\alpha \neq 0\). In this case, \(g(x,y)\) has one critical point, \((0,0)\), if \(\alpha<0\) or \(\beta<0\), otherwise \(g(x,y)\) has five critical points, \((0,0)\) and \((\pm \sqrt{\frac{3ca\beta}{\alpha^2}}, \pm \sqrt{\frac{ca}{\alpha}})\), that could lie in the interior of \({\mathfrak{R}_{RvdP}}_{\omega^2=1}\) since
	
	\begin{equation*}
		\frac{ca\alpha+3ca\beta}{\alpha^2} \leqslant \frac{ca \alpha+3ca \beta+2\sqrt{3c^2a^2 \alpha \beta+\alpha^2 ca (ca-\gamma)}}{\alpha^2}
	\end{equation*}
	
	\noindent for \(\alpha>0\) and \(\beta>0\). As before, \(g(0,0) >0\). At points \((\pm \sqrt{\frac{3ca\beta}{\alpha^2}}, \pm \sqrt{\frac{ca}{\alpha}})\), 
	
	\begin{align*}
		g(x,y)=&ca(ca-\gamma)+ca\alpha \frac{3ca\beta}{\alpha^2}+3ca \beta \frac{ca}{\alpha}- \alpha^2  \frac{3ca\beta}{\alpha^2} \frac{ca}{\alpha} \\
		=&ca(ca-\gamma)+\frac{3c^2a^2\beta}{\alpha}+\frac{3c^2a^2\beta}{\alpha}-\frac{3c^2a^2\beta}{\alpha} \\
		=&ca(ca-\gamma)+\frac{3c^2a^2\beta}{\alpha} \\
		=&\frac{ca}{\alpha^2}[\alpha^2(ca-\gamma)+3ca\alpha \beta]>0. \\
	\end{align*}
	
	On the boundary of \({\mathfrak{R}_{RvdP}}_{\omega^2=1}\), we have seen in Case 3 that \(h(x)\) reaches its global minimum value at \(x=x_{+}\) or at \(x=x_{-}\) because \(r^2 \geqslant \frac{ca \alpha-3ca \beta}{\alpha^2}\). In this case, \(x_{\pm} \in [-r,r]\) since \(r^2 \geqslant \frac{3ca \beta-ca \alpha}{\alpha^2}\). Thus, for \(x \in [-r,r]\), \(h(x)\) has a minimum at \(x=x_{+}\) or at \(x=x_{-}\). Note that \small
	\begin{gather*}
		h(x_{\pm}) = l(r^2) \equiv\frac{-(ca\alpha-3ca \beta-\alpha^2 r^2)^2}{4\alpha^2}+ca(ca-\gamma)+3ca \beta r^2 \\
		=-\frac{\alpha^2}{4}r^4+\frac{ca\alpha+3ca\beta}{2}r^2-\frac{(ca\alpha-3ca\beta)^2}{4\alpha^2}+ca(ca-\gamma).
	\end{gather*}
	
	\noindent \normalsize Since \(3ca\alpha\beta+\alpha^2 (ca-\gamma) > 0\), \(3c^2a^2 \alpha \beta+\alpha^2ca(ca-\gamma)>0\). Observe that at \(r^2=\frac{ca \alpha+3ca \beta\pm 2\sqrt{3c^2a^2 \alpha \beta+\alpha^2ca(ca-\gamma)}}{\alpha^2}\), \(l(r^2)=0\). Hence \(l(r^2) > 0\) given 
	
	\begin{gather*}
		\frac{ca \alpha+3ca \beta-2\sqrt{3c^2a^2 \alpha \beta+\alpha^2ca(ca-\gamma)}}{\alpha^2} < r^2 \\
		< \frac{ca \alpha+3ca \beta+2\sqrt{3c^2a^2 \alpha \beta+\alpha^2ca(ca-\gamma)}}{\alpha^2}.
	\end{gather*}
	
	\noindent Thus, \(h(x)>0\) for all \(x \in [-r,r]\), and therefore \(g(x^*,y^*) > 0\).
	
	Therefore, \({\mathfrak{R}_{RvdP}}_{\omega^2=1} \subset {\mathfrak{C}_{RvdP}}_{\omega^2=1}\) in all four cases.
\end{proof}

\subsection{\label{appB:subsec5}Proof of Lemma \ref{lem4.5}}

\begin{proof}
	The case when \(N=1\) is trivial. We assume Eq.~(\ref{40*}) is true for \(N \geqslant 1\) and choose any row or column to calculate the determinants of \(W_{N+1}\) and \(S_{N+1}\). By switching rows, one can find
	\begin{equation*}
		\begin{split}
			\text{det}[W_{N+1}]=&a(\text{det}[W_{N}])-Nb(\text{det}[S_N]) \\
			=&a(a-b)^{N-1}[a+(N-1)b]-Nb^2(a-b)^{N-1} \\
			=&(a-b)^{N-1}[a^2+(N-1)ab-Nb^2] \\
			=&(a-b)^{N-1}[(a-b)(a+Nb)] \\
			=&(a-b)^{N}(a+Nb),
		\end{split}
	\end{equation*}
	
	\noindent and
	\begin{equation*}
		\begin{split}
			\text{det}[S_{N+1}]=&b(\text{det}[W_{N}])-Nb(\text{det}[S_N]) \\
			=&b(a-b)^{N-1}[a+(N-1)b]-Nb^2(a-b)^{N-1} \\
			=&(a-b)^{N-1}[ab-b^2] \\
			=&b(a-b)^{N}.
		\end{split}
	\end{equation*}
	
	\noindent By mathematical induction, this proof is completed.
\end{proof}

\nocite{*}
\bibliography{Chaos_format}

\end{document}